\documentclass{article}

\usepackage[paperwidth=7in, paperheight=10in, textwidth=5.25in, textheight=8.2in]{geometry}

\usepackage{amsmath}
\usepackage{amsfonts,amsthm,amssymb}
\usepackage{paralist}
\usepackage{nicefrac}
\usepackage{mdframed}
\usepackage{xspace}
\usepackage{algorithm}
\usepackage[noend]{algorithmic}
\usepackage{forloop}
\usepackage{natbib}

\usepackage{graphicx}
\usepackage{subfig} 
\usepackage{float}
\usepackage{microtype}
\usepackage{graphicx}
\usepackage{booktabs} 
\usepackage[hidelinks]{hyperref}
\usepackage{authblk}

\newcommand{\reals}{\ensuremath{\mathbb{R}}}

\newcommand{\prob}[1]{\mathrm{Pr} \left[ #1 \right] } 
\newcommand{\Exp}[2][]{\mathbb{E}_{#1} \left[ #2 \right] } 
\newcommand{\CondExp}[3][]{\mathbb{E}_{#1} \left[ #2  \mid #3 \right] }

\newtheorem{theorem}{Theorem}
\newtheorem{lemma}[theorem]{Lemma}

\newtheorem{corollary}[theorem]{Corollary}
\newtheorem{definition}{Definition}[section]
\newtheorem{proposition}[theorem]{Proposition}
\newtheorem{observation}[theorem]{Observation}
\newtheorem{claim}[theorem]{Claim}

\makeatletter
\newtheorem*{rep@theorem}{\rep@title}
\newcommand{\newreptheorem}[2]{%
	\newenvironment{rep#1}[1]{%
		\def\rep@title{#2 \ref{##1}}%
		\begin{rep@theorem}}%
		{\end{rep@theorem}}}
\makeatother

\newreptheorem{theorem}{Theorem}
\newreptheorem{reduction}{Reduction}

\newcommand{\defcal}[1]{\expandafter\newcommand\csname c#1\endcsname{{\mathcal{#1}}}}
\newcommand{\defbb}[1]{\expandafter\newcommand\csname b#1\endcsname{{\mathbb{#1}}}}
\newcounter{calBbCounter}
\forLoop{1}{26}{calBbCounter}{
	\edef\letter{\Alph{calBbCounter}}
	\expandafter\defcal\letter
	\expandafter\defbb\letter
}

\newcommand{\ground}{\Omega} 
\newcommand{\opt}{\mathit{OPT}}
\newcommand{\gammaapx}{\hat{\gamma}}
\DeclareMathOperator*{\argmax}{arg\,max}

\newcommand{\tr}[1]{Tr \left( #1 \right) }
\newcommand{\trinv}[1]{Tr \left( #1 \right)^{-1} }

\newcommand{\eps}{\varepsilon}
\newcommand{\ie}{{\it i.e.}}

\newcommand{\nnR}{{\bR_{\geq 0}}}
\newcommand{\property}[1]{{(P\ref{#1})}}

\newcommand{\dgreedy}{\textsc{Distorted Greedy}\xspace} 
\newcommand{\sdgreedy}{\textsc{Stochastic Distorted Greedy}\xspace}
\newcommand{\udgreedy}{\textsc{Unconstrained Distorted Greedy}\xspace}
\newcommand{\gsweep}{\textsc{$\gamma$-Sweep}\xspace}

\title{Submodular Maximization Beyond Non-negativity: \\ Guarantees, Fast Algorithms, and 
Applications}

\author[1]{Christopher Harshaw}
\author[2]{Moran Feldman}
\author[3]{\\Justin Ward}
\author[1]{Amin Karbasi}
\affil[1]{Yale University}
\affil[2]{Open University of Israel}
\affil[3]{ Queen Mary University of London}

\date{}

\begin{document}
\maketitle

\begin{abstract}	
It is generally believed that submodular functions{\textemdash}and the more general class of 
$\gamma$-weakly submodular functions{\textemdash}may only be optimized under the 
non-negativity 
assumption $f(S) \geq 0$.
In this paper, we show that once the function is expressed as the difference $f = g - c$, where 
 $g$ is monotone, non-negative, and $\gamma$-weakly submodular and $c$ is non-negative 
modular, then strong approximation guarantees may be obtained.
We present an algorithm for maximizing $g - c$ under a $k$-cardinality constraint which 
produces a random feasible set $S$ such that 
$\Exp{g(S) \! - \! c(S)} \! \geq \! (1  -  e^{-\gamma}   \! - \!  \epsilon) g(\opt) \! - \! c(\opt)$, 
whose running time is $O (\frac{n}{\epsilon} \log^2 \frac{1}{\epsilon})$, i.e., independent of $k$.
We extend these results to the unconstrained setting by describing an algorithm with the same 
approximation guarantees and faster $O(
\frac{n}{\epsilon} \log\frac{1}{\epsilon})$ runtime.
The main techniques underlying our algorithms are two-fold: the use of a surrogate objective which 
varies the relative importance between $g$ and $c$ throughout the algorithm, and a geometric 
sweep 
over possible $\gamma$ values.
Our algorithmic guarantees are complemented by a hardness result showing that no polynomial-time 
algorithm which accesses $g$ through a value oracle can do better.
We empirically demonstrate the success of our algorithms by applying them to experimental 
design on the Boston Housing dataset and directed vertex cover on the Email EU dataset.
\end{abstract}

\section{Introduction}

From summarization and recommendation to clustering and inference, many machine 
learning tasks are inherently discrete. 
Submodularity is an attractive property when designing discrete objective functions, as it encodes a 
natural diminishing returns condition and also comes with an extensive literature on optimization 
techniques. 
For example, submodular optimization techniques have been successfully applied in a wide variety of 
machine learning tasks, including 
sensor placement \citep{krause05near}, 
document summarization \citep{lin2011class},
speech subset selection \citep{wei2013speech}
influence maximization in social networks \citep{kempe03},
information gathering \citep{golovin11},
and
graph-cut based image segmentation \citep{Boykov2001, jegelka2011submodularity},
to name a few.
However, in instances where the objective function is not submodular, existing techniques for 
submodular optimization many perform arbitrarily poorly, motivating the need to study broader 
function classes.
While several notions of approximate submodularity have been studied, the class 
of $\gamma$-weakly submodular functions have (arguably) enjoyed the most practical success.
For example, $\gamma$-weakly submodular optimization techniques have been used in
feature selection \citep{das2011submodular, khanna2017scalable}, anytime linear prediction 
\citep{Hu2016EfficientFG}, interpretation of deep neural networks \citep{elenberg2017}, and high 
dimensional sparse 
regression problems \citep{elenberg2018}.

Here, we study the constrained maximization problem
\begin{equation} \label{eq:main_problem}
\max_{|S| \leq k} g(S) - c(S)
\enspace,
\end{equation}
where $g$ is a non-negative monotone $\gamma$-weakly submodular function and $c$ is a 
non-negative 
modular function. 
Problem~\eqref{eq:main_problem} has various interpretations which may extend the current 
submodular framework to apply to more tasks in machine learning. 
For instance, the modular cost $c$ may be 
added as a penalty to existing submodular maximization problems to encode a cost for each element.
Such a penalty term may play the role of a regularizer or soft constraint in a model.
When $g$ models the revenue of some collection of products $S$ and $c$ models the cost of 
each item, then \eqref{eq:main_problem} corresponds to maximizing profits.

While Problem~\eqref{eq:main_problem} has promising modeling potential, existing optimization 
techniques fail to provide nontrivial approximation guarantees for it. The main reason for that is that most existing 
techniques require the objective function to take only non-negative values, while $g(S) - c(S)$ may 
take both positive and negative values. Moreover, $g(S) - c(S)$ might be non-monotone, and thus, the definition of $\gamma$-weak submodularity does not even apply to it when $\gamma < 1$.

\paragraph{Our Contributions.} We provide several fast algorithms for solving 
Problem~\eqref{eq:main_problem} as well as a matching hardness result and experimental validation of 
our methods. In particular,
\begin{enumerate}
	\item \textbf{Algorithms.} In the case where $\gamma$ is known, we provide a deterministic algorithm 
	which uses $O(nk)$ function evaluations and  
	returns a set $S$ such that $g(S) - c(S) \geq (1 - e^{-\gamma}) g(\opt) - c(\opt)$. If $g$ is regarded 
	as revenue and $c$ as a cost, then this guarantee intuitively states that the algorithm will return a 
	solution whose total profit is at least as much as would be obtained by paying the same cost as the 
	optimal solution while gaining at least a fraction of $(1 - e^{-\gamma})$ out of the revenue of the last solution. 
	We also describe a randomized variant of our algorithm which uses $O(n \log \frac{1}{\epsilon})$ function evaluations
	and has a similar approximation guarantee in expectation, but with an $\epsilon$ additive loss in the 
	approximation factor.
	For the unconstrained setting (when $k=n$) we provide another randomized algorithm which achieves the 
	same approximation guarantee in expectation using only $O(n)$ function evaluations.
	When $\gamma$ is unknown, we give a meta-algorithm for guessing $\gamma$ that loses a 
	$\delta$ additive factor in the approximation ratio and increases the run time by a multiplicative 
	$O(\frac{1}{\delta} \log \frac{1}{\delta})$ factor.
	\item \textbf{Hardness of Approximation.} To complement our algorithms, we provide a matching 
	hardness result which shows that no algorithm which makes polynomially many queries in the value 
	oracle model may do better. To the best of our knowledge, this is the first hardness result of this kind for $\gamma$-weakly submodular 
	functions.
	\item \textbf{Experimental Evaluation.} We demonstrate the effectiveness of our algorithm on 
	experimental design on the Boston Housing dataset and directed vertex cover on the Email EU 
	dataset, both with costs. 
\end{enumerate}

\paragraph{Prior Work} 
The celebrated result of \citet{Nemhauser1978a} showed that the greedy 
algorithm achieves a $(1-1/e)$ approximation for maximizing a nonnegative monotone submodular 
function subject to a cardinality constraint.
\citet{das2011submodular} showed the more general result that the greedy algorithm achieves a 
$(1-e^{-\gamma})$ approximation when $g$ is $\gamma$-weakly submodular.
At the same time, an extensive line of research has lead to the development of algorithms to handle 
non-monotone submodular objectives and/or more complicated constraints (see, e.g.,~\cite{Buchbinder2016,Chekuri2014,Ene2016,Feldman2017,Lee2010,Sviridenko04}).
The $(1-1/e)$ approximation 
was shown to be optimal in the value oracle model \cite{Nemhauser1978}, but until this work, no 
stronger hardness result was known for constrained $\gamma$-weakly submodular maximization.
 The problem of maximizing $g + \ell$ for non-negative monotone submodular $g$ and an (arbitrary) 
 modular function $\ell$ under cardinality constraints 
was first considered in \cite{Sviridenko2017}, who gave a randomized polynomial time algorithm which 
outputs a set $S$ such that $g(S) + \ell(S) \geq (1 - 1/e) g(\opt) + \ell(\opt)$, where $\opt$ is the 
optimal set. 
This approximation was shown to be optimal in the value oracle model via a reduction from submodular 
maximization with bounded curvature. However, the algorithm of \citet{Sviridenko2017} is of mainly 
theoretical interest, as it requires continuous optimization of the multilinear extension and an 
expensive routine to guess the contribution of $\opt$ to the modular term, yielding it practically 
intractable. 
\citet{feldman2019} suggested the idea of using a surrogate objective that varies with time, and 
showed that this idea removes the need for the guessing step. However, the algorithm 
of~\cite{feldman2019} still requires expensive sampling as it is based on the multilinear extension.
Moreover, neither of these approaches can currently handle $\gamma$-weakly submodular functions, as 
optimization routines that go through their multilinear extensions have not yet been developed. 

\paragraph{Organization} The remainder of the paper is organized as follows. Preliminary definitions 
are given in Section~\ref{sec:preliminaries}. The algorithms we present for solving 
Problem~\eqref{eq:main_problem} 
are presented in Section~\ref{sec:algorithms}. The hardness result is stated in 
Section~\ref{sec:hardness}. Applications, experimental set-up, and experimental results are discussed in 
Section~\ref{sec:experiments}. Finally, we conclude with a discussion in Section~\ref{sec:conclusion}. 

\section{Preliminaries} \label{sec:preliminaries}
Let $\ground$ be a ground set of size $n$. 
For a real-valued set function $g\colon 2^\ground \rightarrow \reals$, we write the marginal gain of adding 
an element $e$ to a set $A$ as $g( e \mid S) \triangleq g(S \cup \{e\}) - g(S)$. We say that $g$ is 
\emph{monotone} if $g(A) \leq g(B)$ for all $A \subseteq B$, and say that $g$ is 
\emph{submodular} if for all sets $A \subseteq B \subseteq \ground$ and element $e \notin B$,
\begin{equation} \label{eq:submodular_def}
g(e \mid A) \geq g( e \mid B) \enspace.
\end{equation}
When $g$ is interpreted as a utility function, \eqref{eq:submodular_def} encodes a natural diminishing 
returns condition in the sense that the marginal gain of adding an element decreases as the current set 
grows larger. An equivalent definition is that $\sum_{e \in B} g( e \mid A) \geq g(A \cup B) - g(A)$, 
which allows for the following natural extension. A monotone set function $g$ 
is \emph{$\gamma$-weakly submodular} for $\gamma \in (0,1]$ if
\begin{equation} \label{eq:weakly_submodular_def}
\sum_{e \in B \setminus A} g( e \mid A) \geq \gamma \left( g(A \cup B) - g(A) \right)
\end{equation}
holds for all $A \subseteq B$. In this case,  $\gamma$ is referred to as the \emph{submodularity ratio}. 
Intuitively, such a function $g$ may not have strictly diminishing returns, but the increase in the returns 
is bounded by the marginals.
Note that $g$ is submodular if and only if it is $\gamma$-weakly submodular with $\gamma=1$. 
A real-valued set function $c\colon 2^\ground \rightarrow \reals$ is \emph{modular} if 
\eqref{eq:submodular_def} holds with equality. 
A modular function may always be written in terms of 
coefficients as $c(S) = \sum_{e \in S} c_e$ and is non-negative if and only if all of its coefficients are 
non-negative.

Our algorithms are specified in the \emph{value oracle model}, namely under the assumption that there 
is an 
oracle that, 
given a set $S \subseteq \ground$, returns the value $g(S)$. As is standard, we analyze the run time 
complexity of these algorithms in terms of the number of function evaluations they require.

\section{Algorithms} \label{sec:algorithms}
In this section, we present a suite of fast algorithms for solving Problem~\ref{eq:main_problem}. 
The main idea behind each of these algorithms is to optimize a surrogate objective, which changes 
throughout the algorithm, preventing us from getting stuck in poor local optima. 
Further computational speed ups are obtained by randomized sub-sampling of the ground set.\footnote{We note that these two techniques can be traced back to the works of~\cite{feldman2019} and~\cite{Mirzasoleiman2015}, respectively.}
The first algorithms we present assume knowledge of the weak submodularity parameter 
$\gamma$. However, $\gamma$ is rarely known in practice (unless it is equal to $1$), and thus, we show in Section~\ref{sec:gamma_sweep} how to adapt these algorithms for the case of unknown $\gamma$.

To motivate the distorted objective we use, let us describe a way in which the greedy algorithm may fail.
Suppose there is a ``bad element'' $b \in \ground$ which has the highest overall gain, $g(b) - c_b$ 
and 
so 
is added to the solution set; however, once added, the marginal gain of all remaining elements 
drops below the corresponding costs, and so the greedy algorithm terminates.
This outcome is suboptimal when there are other elements $e$ that, although their overall marginal 
gain $g(e \mid S) - c_e$ is lower, have much higher ratio between the marginal utility $g(e \mid S)$ and 
the
cost $c_e$ (see Appendix~\ref{sec:greedy_performs_poorly} for an explicit construction).

To avoid this type of situation, we design a distorted objective which initially places higher relative 
importance on the modular cost term $c$, and gradually increases the relative importance of the utility 
$g$ as the algorithm progresses. 
Our analysis relies on two functions: $\Phi$, the distorted objective, and $\Psi$, an 
important quantity in analyzing the trajectory of $\Phi$. 
Let $k$ denote the cardinality 
constraint, then for any $i=0,1,\dotsc, k$ and any set $T$, we define
\begin{equation*}
\Phi_i(T) \triangleq \left( 1 - \frac{\gamma}{k} \right)^{k-i} g(T) - c(T)
\enspace.
\end{equation*}
Additionally, for any iteration $i=0,1, \dots, k-1$ of our algorithm, a set $T \subseteq \ground$, and an element $e \in \ground$, let
\begin{equation*}
\Psi_i(T, e) \triangleq \max \left\{ 0,  \left(1 - \frac{\gamma}{k} \right)^{k - (i+1)} g(e \mid T ) - c_e \right\}
\enspace.
\end{equation*}


\subsection{Distorted Greedy}
Our first algorithm, \dgreedy, is presented as Algorithm~\ref{alg:distorted_greedy}.
At each iteration, this algorithm chooses an element $e_i$ maximizing the increase in the distorted objective.
The algorithm then only accepts $e_i$ if it positively contributes to the distorted objective. 

\begin{algorithm}[ht]
	\caption{\dgreedy}
	\label{alg:distorted_greedy}
	\begin{algorithmic}
		\STATE {\bfseries Input:} utility $g$, weak $\gamma$, cost $c$, cardinality $k$
		\STATE Initialize $S_0 \gets \varnothing$
		\FOR{$i=0$ {\bfseries to} $k-1$}
		\STATE $e_i \gets \argmax_{e \in \ground} \big\{\big( 1 - \frac{\gamma}{k} \big)^{k-(i+1)} g(e \mid 
		S_i) - c_e\big\}$
		\IF{$\left( 1 - \frac{\gamma}{k} \right)^{k-(i+1)} g(e_i \mid S_i) - c_{e_i} > 0$}
		\STATE $S_{i+1} \gets S_i \cup \{e_i\}$
		\ELSE 
		\STATE $S_{i+1} \gets S_i $
		\ENDIF
		\ENDFOR
		\STATE {\bfseries Return} $S_k$
	\end{algorithmic}
\end{algorithm}

The analysis consists 
mainly of two lemmas. First, Lemma~\ref{lem:dist_gain_lb} shows that the marginal gain in the distorted 
objective is lower bounded by a term involving $\Psi$. This fact relies on 
the non-negativity of $c$ and the rejection step in the algorithm.
\begin{lemma} \label{lem:dist_gain_lb}
	In each iteration of \dgreedy,
	\[
	\Phi_{i+1}(S_{i+1}) - \Phi_i(S_i) 
	=  \Psi_i(S_i, e_i) + \frac{\gamma}{k}\left( 1 - \frac{\gamma}{k} \right)^{k-(i+1)}  g(S_i) 
	\enspace.
	\]
\end{lemma}
\begin{proof}
	By expanding the definition of $\Phi$ and rearranging, we get
	\begin{align*}
	\Phi_{i+1}\mspace{-3mu}&\mspace{3mu}(S_{i+1}) - \Phi_i(S_i) \\
	&= \left( 1 - \frac{\gamma}{k} \right)^{k-(i+1)} g(S_{i+1}) - c(S_{i+1}) - \left( 1 - \frac{\gamma}{k} 
	\right)^{k-i} g(S_i) + c(S_i) \\
	&= \left( 1 - \frac{\gamma}{k} \right)^{k-(i+1)} g(S_{i+1}) - c(S_{i+1}) - \left( 1 - \frac{\gamma}{k} 
	\right)^{k-(i+1)} \left(1 - \frac{\gamma}{k} \right) g(S_i) + c(S_i) \\
	&= \left( 1 - \frac{\gamma}{k} \right)^{k-(i+1)} \left[  g(S_{i+1}) - g(S_i) \right]
	- \left[ c(S_{i+1}) - c(S_i) \right] + \frac{\gamma}{k}\left( 1 - \frac{\gamma}{k} 
	\right)^{k-(i+1)} g(S_i) \enspace.
	\end{align*}
	Now let us consider two cases. First, suppose that the if statement in {\dgreedy} passes, which 
	means that
	$\Psi_i(S_i, e_i) = \left(1 - \frac{\gamma}{k} \right)^{k-(i+1)} g(e_i \mid S_i ) - c_{e_i}  > 0$ and that 
	$e_i$ is 
	added to the solution set. By the non-negativity of $c$, we can deduce in this case that $e_i 
	\notin S_i$, and thus, $g(S_{i+1}) - g(S_i) = g(e_i \mid S_i )$ and $c(S_{i+1}) - c(S_i) = c_{e_i}$. Hence,
	\begin{align*} \Phi_{i+1}(S_{i+1}) - \Phi_i(S_i)  
	={} &  \left(1 - \frac{\gamma}{k} \right)^{k-(i+1)} g(e_i \mid S_i ) - c_{e_i} + 
	\frac{\gamma}{k}\left( 1 - \frac{\gamma}{k} 
	\right)^{k-(i+1)} g(S_{i})\\
	={} & \Psi_i(S_i, e_i) + \frac{\gamma}{k}\left( 1 - \frac{\gamma}{k} 
	\right)^{k-(i+1)} g(S_i) 
	\enspace.
	\end{align*}
	Next, suppose that the if statement in $\dgreedy$ does not pass, which means that $\Psi_i(S_i, e_i) 
	= 0 \geq 
	\left(1 - \frac{\gamma}{k} \right)^{k-(i+1)} g(e_i \mid S_i ) - c_{e_i}$ and the algorithm does not add 
	$e_i$ to its solution. In particular, $S_{i+1} = S_i$, and thus, $g(S_{i+1}) - g(S_i) = 0$ and $c(S_{i+1}) 
	- c(S_i)  = 0$. In this case,
	\begin{align*} \Phi_{i+1}(S_{i+1}) - \Phi_i(S_i)  ={}& 0 +  \frac{\gamma}{k}\left( 1 - \frac{\gamma}{k} 
	\right)^{k-(i+1)} g(S_i)\\ ={}& \Psi_i(S_i, e_i) +  
	\frac{\gamma}{k}\left( 1 - \frac{\gamma}{k} 
	\right)^{k-(i+1)} g(S_i)\enspace.\qedhere\end{align*}
\end{proof}

The second lemma shows that the marginal gain in the distorted objective is sufficiently large to 
ensure the desired approximation guarantees. This fact relies on the monotonicity and
$\gamma$-weak submodularity of $g$.
\begin{lemma} \label{lem:psi_lb_dgreedy}
	In each iteration of \dgreedy,
	\[
	\Psi_i(S_i, e_i) \geq \frac{\gamma}{k} \left(1 - \frac{\gamma}{k} \right)^{k - (i+1)}  
	\left[ g(\opt) - g(S_i) \right] - \frac{1}{k} c(\opt) \enspace.
	\]
\end{lemma}
\begin{proof} 
	Observe that
	\begin{align*}
	k &\cdot \Psi_i(S_i, e_i) \\
	&= k \cdot \max_{e \in \ground} \left\{ 0, \left(1 - \frac{\gamma}{k} \right)^{k - (i + 1)} g(e \mid S_i) - 
	c_{e} 
	\right\} &\text{(definitions of $\Psi$ and $e_i$)}\\
	&\geq |OPT| \cdot \max_{e \in \ground} \left\{ 0, \left(1 - \frac{\gamma}{k} \right)^{k - (i + 1)} g(e \mid 
	S_i) - c_{e}\right\} &\text{($\lvert \opt \rvert \leq k$)}\\
	&\geq |OPT| \cdot \max_{e \in \opt} \left\{\left(1 - \frac{\gamma}{k} \right)^{k - (i + 1)} g(e \mid S_i) - 
	c_e \right\} 
	&\text{(restricting maximization)}\\
	&\geq \sum_{e \in \opt} \left[ \left(1 - \frac{\gamma}{k} \right)^{k - (i + 1)} g(e \mid 
	S_i) - c_e \right] &\text{(averaging argument)}\\
	&= \left(1 - \frac{\gamma}{k} \right)^{k - (i + 1)} \sum_{e \in \opt} g(e \mid S_i) 
	- c(\opt) &\\
	&\geq \gamma \left(1 - \frac{\gamma}{k} \right)^{k - (i + 1)} \left[ g(\opt) - g(S_i) \right] - c(\opt) 
	\enspace.
	&\text{($\gamma$-weak submodularity)} \tag*{\qed}
	\end{align*}
	\let\qed\relax
\end{proof}

Using these two lemmas, we present an approximation guarantee for \dgreedy.
\begin{mdframed}[nobreak=true]
\begin{theorem} \label{thm:distorted_greedy}
	\dgreedy makes $O(nk)$ evaluations of $g$ and returns a set $R$ of size at most $k$ with
	\[ g(R) - c(R) \geq \left( 1 - e^{-\gamma} \right) g(\opt) - c(\opt) \enspace. \]
\end{theorem}
\end{mdframed}
\begin{proof}
	Since $c$ is modular and $g$ is non-negative, the definition of $\Phi$ gives
	\[
	\Phi_0(S_0) = \left( 1 - \frac{\gamma}{k} \right)^{k} g(\varnothing) - c(\varnothing) 
	= \left( 1 - \frac{\gamma}{k} \right)^{k}g(\varnothing)
	 \geq 0 \enspace.
	\]
and
	\[
	\Phi_k \left( S_k \right) = \left( 1 - \frac{\gamma}{k} \right)^{0} g(S_k) - c(S_k) = g(S_k) - c(S_k) 
	\enspace.
	\]
	Using this and the fact that the returned set $R$ is in fact $S_k$, we get
	\begin{equation} \label{eq:telescoping_sum1}
	g(R) - c(R) 
	\geq \Phi_k (S_k) - \Phi_0 (S_0) 
	= \sum_{i=0}^{k-1} \Phi_{i+1}(S_{i+1}) - \Phi_i(S_i)  \enspace.
	\end{equation}
	Applying Lemmas~\ref{lem:dist_gain_lb} and \ref{lem:psi_lb_dgreedy}, respectively, we have
	\begin{align*}
	\Phi_{i+1}(S_{i+1}) - \Phi_i(S_i) 
	&= \Psi_i(S_i, e_i) + \frac{\gamma}{k}\left( 1 - \frac{\gamma}{k} 
	\right)^{k-(i+1)} g(S_i) \\
	&\geq \frac{\gamma}{k} \left(1 - \frac{\gamma}{k} \right)^{k-(i+1)}  \left[ g(\opt) - g(S_i) 
	\right] \\
	&\quad
	 - \frac{1}{k} c(\opt)  + \frac{\gamma}{k}\left( 1 - \frac{\gamma}{k} 
	\right)^{k-(i+1)} g(S_i) \\
	&= 
	\frac{\gamma}{k} \left(1 - \frac{\gamma}{k} \right)^{k-(i+1)} g(\opt) - \frac{1}{k} c(\opt)
	\enspace.
	\end{align*}
	Finally, plugging this bound into 
	\eqref{eq:telescoping_sum1} yields
	\begin{align*}
	g(R) - c(R) 
	&\geq  \sum_{i=0}^{k-1} \left[ \frac{\gamma}{k} \left(1 - \frac{\gamma}{k} \right)^{k-(i+1)} g(\opt) 
	- \frac{1}{k} c(\opt) \right] \\
	&=  \left[ \frac{\gamma}{k} \sum_{i=0}^{k-1} \left(1 - \frac{\gamma}{k} \right)^{i} \right] g(\opt) - 
	c(\opt) \\
	&= \left(1 - \left(1 - \frac{\gamma}{k} \right)^{k} \right) g(\opt) - c(\opt) \\
	&\geq \left( 1 - e^{-\gamma} \right) g(\opt) - c(\opt)
	\enspace.
	\qedhere
	\end{align*}
\end{proof}

\subsection{Stochastic Distorted Greedy}
Our second algorithm, \sdgreedy, is presented as Algorithm~\ref{alg:stochastic_distorted_greedy}.
It uses the same distorted objective as \dgreedy, but enjoys an asymptotically faster run time due to 
sampling techniques of \cite{Mirzasoleiman2015}. Instead of optimizing over the entire ground set at 
each iteration, {\sdgreedy} optimizes over a random sample $B_i \subseteq \ground$ of size $O \left( 
\frac{n}{k} \log \frac{1}{\epsilon} \right)$. This sampling procedure ensures that sufficient potential gain 
occurs \emph{in expectation}, which is true for the following reason. If 
the sample size is sufficiently large, then  $B_i$ contains at least one element of $\opt$ with 
high probability. Conditioned on this (high probability) event, choosing the element with the maximum 
potential gain is at least as good as choosing an average element from $\opt$. 

\begin{algorithm}[h!]
	\caption{\sdgreedy}
	\label{alg:stochastic_distorted_greedy}
	\begin{algorithmic}
		\STATE {\bfseries Input:} utility $g$, weak $\gamma$, cost $c$, cardinality $k$, error 
		$\epsilon$
		\STATE Initialize $S_0 \gets \varnothing$, $s \gets \lceil \frac{n}{k} \log(\frac{1}{\epsilon}) \rceil$
		\FOR{$i=0$ {\bfseries to} $k-1$}
		\STATE $B_i \gets$ sample $s$ elements uniformly and independently from $\ground$
		\STATE $e_i \gets \argmax_{e \in B_i} \big\{\big( 1 - \frac{\gamma}{k} \big)^{k-(i+1)} g(e \mid S_i) - 
		c_e\big\}$
		\IF{$\left( 1 - \frac{\gamma}{k} \right)^{k-(i+1)} g(e_i \mid S_i) - c_{e_i} > 0$}
		\STATE $S_{i+1} \gets S_i \cup \{e_i\}$
		\ELSE 
		\STATE $S_{i+1} \gets S_i $
		\ENDIF
		\ENDFOR
		\STATE {\bfseries Return} $S_k$
	\end{algorithmic}
\end{algorithm}

The next three lemmas formalize this idea and are analogous to Lemma~2 in \cite{Mirzasoleiman2015}.
The first step is to show that an element of $\opt$ is likely to appear in the sample $B_i$.
\begin{lemma} \label{lem:sampling_lem}
	In each step $0 \leq i \leq k-1$ of \sdgreedy,
	\[ \prob{B_i \cap \opt \neq \varnothing } \geq \left( 1 - \epsilon \right) \frac{ \lvert \opt 
		\rvert}{k}\enspace.\]
\end{lemma}
\begin{proof}
	\begin{equation*}
	\prob{B_i \cap \opt = \varnothing}
	\leq \left(1 - \frac{ \lvert \opt \rvert}{n} \right)^s 
	\leq e^{-s \frac{ \lvert \opt \rvert }{n}} 
	= e^{-\frac{s k}{n} \frac{ \lvert \opt \rvert}{k} }\enspace,
	\end{equation*}
	where we used the known inequality $1 - x \leq e^{-x}$. Thus,
	\begin{equation*}
	\prob {B_i \cap \opt \neq \varnothing}
	\geq 1 - e^{-\frac{s k}{n} \frac{ \lvert \opt \rvert}{k} } 
	\geq \left( 1 - e^{-\frac{s k}{n}} \right) \frac{ \lvert \opt \rvert}{k} 
	\geq \left( 1 - \epsilon \right) \frac{ \lvert \opt \rvert}{k} \enspace,
	\end{equation*}
	where the second inequality follows from $1 - e^{-ax} \geq (1 - e^{-a})x$ for $x \in [0,1]$ and $a > 0$, and the 
	last 
	inequality follows from the choice of sample size $s = \lceil \frac{n}{k} \log \frac{1}{\epsilon} \rceil$.
\end{proof}

Conditioned on the fact that at least one element of $\opt$ was sampled, the following lemma shows 
that sufficient potential gain is made.
\begin{lemma} \label{lem:exp_psi_conditional_bound}
	In each step $0 \leq i \leq k-1$ of \sdgreedy,
	\begin{align*}
	\CondExp[e_i]{\Psi_i(S_i, e_i)}{ S_i, \ B_i \cap \opt \neq \varnothing} 
	&\geq \frac{\gamma}{\lvert \opt \rvert} \left( 1 - \frac{\gamma}{k} \right)^{k- (i+1)} \left[ g(\opt) -g(S_i) 
	\right] \\
	&\quad - \frac{1}{\lvert \opt \rvert} c(\opt) \enspace.
	\end{align*}
\end{lemma}
\begin{proof}
	Throughout the proof, all expectations are 
	conditioned on the 
	current set $S_i$ and the event that $B_i \cap \opt \neq \varnothing$, as in the statement of the 
	lemma. For 
	convenience, we drop the notations of these conditionals from the calculations below.
	\begin{align*}
	\Exp[e_i]{\Psi_i(S_i, e_i)}
	&= \Exp{\max_{e \in B_i} \Psi_i(S_i, e_i) } 
	&\text{(definition of $e_i$)}\\
	&\geq \Exp{ \max_{e \in B_i \cap \opt } \Psi_i(S_i, e_i) }
	&\text{(restricting max)} \\
	&\geq \Exp{ \max_{e \in B_i \cap \opt } \left\{\left( 1 - \frac{\gamma}{k}\right)^{k-(i+1)} g(e \mid S_i) - 
		c_{e} \right\} }
	\enspace.
	&\text{(definition of $\Psi$)}
	\end{align*}
	We now note that $B_i \cap OPT$ is a subset of $OPT$ that contains every element of $OPT$ with 
	the same probability. Moreover, this is true also conditioned on $B_i \cap OPT \neq \varnothing$. 
	Thus, picking the best element from $B_i \cap OPT$ (when this set is not-empty) achieves gain at 
	least as large as picking a random element from $B_i \cap OPT$, which is identical to picking a 
	random element from $OPT$. Plugging this observation into the previous inequality, we get
	\begin{align*}
	\Exp[e_i]{\Psi_i(S_i, e_i)}
	&\geq \frac{1}{\lvert \opt \rvert} \sum_{e \in \opt} \left[ \left( 1 - \frac{\gamma}{k}\right)^{k-(i+1)}  
	g(e \mid S_i) - c_{e} \right]\\
	&= \frac{1}{\lvert \opt \rvert} \left( 1 - \frac{\gamma}{k}\right)^{k-(i+1)} \sum_{e \in \opt} g(e \mid S_i) - 
	\frac{1}{\lvert \opt \rvert} c(\opt) \\
	&\geq \frac{\gamma}{\lvert \opt \rvert} \left( 1 - \frac{\gamma}{k} \right)^{k-(i+1)} 
	\left[ g(\opt \cup S_i) - g(S_i) \right] 
	- \frac{1}{\lvert \opt \rvert} c(\opt) \\
	&\geq \frac{\gamma}{\lvert \opt \rvert} \left( 1 - \frac{\gamma}{k} \right)^{k-(i+1)} 
	\left[ g(\opt) - g(S_i) \right] 
	- \frac{1}{\lvert \opt \rvert} c(\opt)  \enspace,
	\end{align*}
	where the last two inequalities follows from the $\gamma$-weak submodularity and monotonicity of 
	$g$, respectively.
\end{proof}

The next lemma combines the previous two to show that sufficient gain of the distorted objective 
occurs at each iteration. 
\begin{lemma} \label{lem:psi_lb_sdgreedy}
	In each step of 
	\sdgreedy,
	\[
	\Exp{\Psi_i(S_i, e_i)} 
	\geq (1 - \epsilon) \bigg( \frac{\gamma}{k} \left( 1 - \frac{\gamma}{k} \right)^{k-(i+1)} 
	\big[ g(\opt) 	- \Exp{g(S_i)} \big] 
	- \frac{1}{k} c(\opt) \bigg) \enspace. 
	\]
\end{lemma}
\begin{proof}
	By the law of iterated expectation and the non-negativity of $\Psi$,
	\begin{align*}
	&\CondExp[e_i]{\Psi_i(S_i, e_i)}{S_i} \\
	=& \CondExp[e_i]{\Psi_i(S_i, e_i)}{S_i, \ B_i \cap \opt \neq \varnothing} \prob{B_i \cap \opt \neq 
		\varnothing} \\
	& \quad +  \CondExp[e_i]{\Psi_i(S_i, e_i)}{ S_i , \ B_i \cap \opt = \varnothing} \prob{B_i \cap \opt =  
	\varnothing}\\
	\geq{} & \CondExp[e_i]{\Psi_i(S_i, e_i)}{S_i , \ B_i \cap \opt \neq \varnothing} \prob{B_i \cap \opt \neq 
		\varnothing}\\
	\geq{} & \left(\frac{\gamma}{\lvert \opt \rvert} \left( 1 - \frac{\gamma}{k} \right)^{k-(i+1)} \left[ g(\opt) 
	- 
	g(S_i) \right] 
	- \frac{1}{\lvert \opt \rvert} c(\opt) \right) 
	\left( \left( 1 - \epsilon \right) \frac{ \lvert \opt 	\rvert}{k} \right) 
	\\
	={} & (1 - \epsilon) \left( \frac{\gamma}{k} \left( 1 - \frac{\gamma}{k} \right)^{k-(i+1)}
	\left[ g(\opt) - g(S_i) 	\right] - \frac{1}{k} c(\opt) \right)
	\enspace,
	\end{align*}
	where the second inequality holds by Lemmas~\ref{lem:sampling_lem} 
	and~\ref{lem:exp_psi_conditional_bound}.
	The lemma now follows since the law of iterated expectations also implies
	$\Exp{\Psi_i(S_i, e_i)} = \Exp[S_i]{ \CondExp[e_i]{ \Psi_i(S_i, e_i)} {S_i}  }$.
\end{proof}

Using the previous lemmas, we can now prove the approximation guarantees of 
\sdgreedy.
\begin{mdframed}[nobreak=true]
\begin{theorem} \label{thm:stochastic_distorted_greedy}
	\sdgreedy uses $O(n \log \frac{1}{\epsilon})$ evaluations 
	of $g$ and returns a set $R$ with
	\[
	\Exp{g(R)  - c(R)} \geq \left( 1  - e^{-\gamma}  - \epsilon \right) g(\opt) - c(\opt) \enspace.
	\]
\end{theorem}
\end{mdframed}

\begin{proof} 
	As discussed in the proof of Theorem~\ref{thm:distorted_greedy}, we have that
	\begin{equation} \label{eq:telescoping_sum2}
	\Exp{g(R) - c(R)} \geq \Exp{\Phi_k (S_k) - \Phi_0 (S_0)} = \sum_{i=0}^{k-1}  \Exp{\Phi_{i+1}(S_{i+1}) - 
		\Phi_i(S_i)} \enspace,
	\end{equation}
	and so it is enough to lower bound each term in the rightmost side. To this end, we apply 
	Lemma~\ref{lem:dist_gain_lb} and Lemma~\ref{lem:psi_lb_sdgreedy} to obtain
	{\allowdisplaybreaks
	\begin{align*}
	\Exp{ \Phi_{i+1}(S_{i+1}) - \Phi_i(S_i) } 
	\mspace{-12mu}&\mspace{12mu}\geq 
	 \Exp{\Psi_i(S_i, e_i)} +  \frac{\gamma}{k}\left( 1 - 
	\frac{\gamma}{k} 
	\right)^{k-(i+1)} \Exp{g(S_i)}\\
	&\geq  (1 - \epsilon)
	\left( \frac{\gamma}{k} \left( 1 - \frac{\gamma}{k} \right)^{k-(i+1)} \left[ g(\opt) - \Exp{g(S_i)} \right] 
	+ \frac{1}{k} c(\opt) \right) \\
	&\quad +  \frac{\gamma}{k}\left( 1 - \frac{\gamma}{k} 
	\right)^{k-(i+1)}  \Exp{g(S_i)}\\
	&= (1 - \epsilon) \left( \frac{\gamma}{k} \left( 1 - \frac{\gamma}{k} \right)^{k-(i+1)} g(\opt) 
	- \frac{1 }{k} c(\opt) \right) \\
	& \quad +  \epsilon \cdot \frac{\gamma}{k} \left( 1 - \frac{\gamma}{k} 	\right)^{k-(i+1)} \Exp{g(S_i)}\\
	&\geq (1 - \epsilon) \left( \frac{\gamma}{k} \left( 1 - \frac{\gamma}{k} \right)^{k-(i+1)} g(\opt) 
	- \frac{1}{k} c(\opt) \right) \enspace,
	\end{align*}
	}
	where the last inequality followed from non-negativity of $g$. Plugging this bound into 
	\eqref{eq:telescoping_sum2} yields
	\begin{align*}
	\Exp{g(R) - c(R)} 
	&\geq  (1 - \epsilon) \sum_{i=0}^{k-1} \left[ \frac{\gamma}{k} \left(1 - \frac{\gamma}{k} \right)^{k-(i+1)} 
	g(\opt) - \frac{1}{k} c(\opt) \right] \\
	&= (1 - \epsilon) \left[ \frac{\gamma}{k} \sum_{i=0}^{k-1} \left( 1 - \frac{\gamma}{k} \right)^i \right] 
	g(\opt) - (1 - \epsilon) c(\opt) \\
	&\geq (1 - \epsilon) \left( 1 - e^{-\gamma} \right) g(\opt) - c(\opt) \\
	&= \left( 1 - e^{-\gamma} - \alpha \epsilon \right) g(\opt) -  c(\opt) \enspace,
	\end{align*}
	where the second inequality follows from non-negativity of $g$ and $c$, and $\alpha = 1 - 
	e^{-\gamma} \leq 0.65$.
	
	To bound the number of function evaluations used by {\sdgreedy}, observe that 
	this algorithm has $k$ rounds, each requiring $s = \lceil \frac{n}{k} \log \frac{1}{\epsilon} \rceil $ 
	function 
	evaluations. Thus, the total number of function evaluations is 
	$k \times \lceil \frac{n}{k} \log \frac{1}{\epsilon} \rceil  = O(n \log \frac{1}{\epsilon})$.
\end{proof}

\subsection{Unconstrained Distorted Greedy}
In this section, we present \udgreedy, an algorithm for the unconstrained setting (i.e., $k = n$), 
listed as Algorithm~\ref{alg:unconstrained_distorted_greedy}. {\udgreedy} samples a \emph{single} 
random element at 
each iteration, and adds it to the current solution if the potential gain is sufficiently large. Note that this
algorithm is faster than the previous two, as it requires only $O(n)$ evaluations of $g$. 

\begin{algorithm}[ht]
	\caption{\udgreedy}
	\label{alg:unconstrained_distorted_greedy}
	\begin{algorithmic}
		\STATE {\bfseries Input:} utility $g$, weak $\gamma$, cost $c$, cardinality $k$
		\STATE Initialize $S_0 \gets \varnothing$
		\FOR{$i=0$ {\bfseries to} $n-1$}
		\STATE $e_i \gets$ sample uniformly from $\ground$
		\IF{$\left( 1 - \frac{\gamma}{n} \right)^{n - (i+1)} g(e_i \mid S_i) - c_{e_i} > 0$}
		\STATE $S_{i+1} \gets S_i \cup \{e_i\}$
		\ELSE 
		\STATE $S_{i+1} \gets S_i $
		\ENDIF
		\ENDFOR
		\STATE {\bfseries Return} $S_n$
	\end{algorithmic}
\end{algorithm}

Like  \dgreedy and \sdgreedy, \udgreedy relies on the distorted objective and the heart of the 
analysis is showing that the increase of this distorted objective is sufficiently large in each iteration.
However, the argument in the analysis is different.
Our analysis of the previous algorithms argued that ``the best element is better than an 
average element'', while the analysis of \udgreedy works with that average directly.
This allows for significantly fewer evaluations of $g$ required by the algorithm.
\begin{lemma} \label{lem:psi_lb_udgreedy}
	In each step of \udgreedy,
	\[
	\Exp{\Psi_i(S_i, e_i)} \geq
	\frac{\gamma}{n} \left( 1 - \frac{\gamma}{n} \right)^{n - (i+1)} \big[ g(\opt) 
	- \Exp{g(S_i)} \big]  - \frac{1}{n} c(\opt) \enspace.
	\]
\end{lemma}
\begin{proof}
	We begin by analyzing the conditional expectation
	\begin{align*}
	&\CondExp[e_i]{\Psi_i(S_i, e_i)}{S_i} \\
	&= \frac{1}{n} \sum_{e \in \ground} \Psi_i(S_i, e) \\
	&\geq \frac{1}{n} \sum_{e \in \opt} \Psi_i(S_i, e) &\text{(non-negativity of $\Psi$)}\\
	&= \frac{1}{n} \sum_{e \in \opt} \max \left\{ 0, \left( 1 - \frac{\gamma}{n} \right)^{n - (i+1)} g(e \mid S_i) 
	- 
	c_e \right\}
	&\text{(by definition of $\Psi$)} \\
	&\geq \frac{1}{n} \sum_{e \in \opt} \left\{\left(1 - \frac{\gamma}{n} \right)^{n-(i+1)} g(e \mid S_i) - 
	c_e\right\} \\
	&= \frac{1}{n} \left(1 - \frac{\gamma}{n} \right)^{n-(i+1)} \sum_{e \in \opt} g(e \mid S_i) - 
	\frac{1}{n} c(\opt) &\text{(linearity of $c$)}\\
	&\geq \frac{\gamma}{n} \left(1 - \frac{\gamma}{n} \right)^{n-(i+1)} \left[ g(\opt \cup S_i) - g(S_i) \right] 
	- \frac{1}{n}c(\opt) 
	&\text{($\gamma$-weakly submodular)} \\
	&\geq \frac{\gamma}{n} \left(1 - \frac{\gamma}{n} \right)^{n-(i+1)} \left[ g(\opt) - g(S_i) \right] - 
	\frac{1}{n}c(\opt) 
	&\text{(monotonicity of $g$)}\enspace. 
	\end{align*}
	The lemma now follows by the law of iterated expectations.
\end{proof}

In the same way that Theorem~\ref{thm:distorted_greedy} follows from 
Lemma~\ref{lem:psi_lb_dgreedy}, the next theorem follows from Lemma~\ref{lem:psi_lb_udgreedy}, and 
so we omit its proof.
\begin{mdframed}[nobreak=true]
\begin{theorem} \label{thm:udgreedy}
	{\udgreedy} requires $O(n)$ function evaluations and outputs a set 
	$R$ such that
	\[ 
	\Exp{g(R) - c(R)} \geq (1 - e^{-\gamma})g(\opt) - c(\opt) \enspace.
	\]
\end{theorem}
\end{mdframed}

\subsection{Guessing Gamma: A Geometric Sweep} \label{sec:gamma_sweep}
The previously described algorithms required knowledge of the submodularity ratio $\gamma$. 
However, it is very rare that the precise value of $\gamma$ is known in practice{\textemdash}unless 
$g$ is submodular, in which case $\gamma=1$. 
Oftentimes, $\gamma$ 
is data dependent and only a crude lower bound $L \leq \gamma$ is known. In this section, 
we describe a meta algorithm that ``guesses'' the value of $\gamma$. \gsweep , listed as 
Algorithm~\ref{alg:gamma_sweep},
runs a maximization algorithm $\mathcal{A}$ as a subroutine with a geometrically decreasing 
sequence of ``guesses'' $\gamma^{(k)}$ for $k=0, 1, \dotsc, \lceil \frac{1}{\delta} \log \frac{1}{\delta} 
\rceil$. The best set obtained by this procedure is guaranteed to have nearly as good 
approximation guarantees as when the true submodularity ratio $\gamma$ is known exactly.  
Moreover, fewer guesses are required if a good lower bound $L \leq \gamma$ is known, which is true 
for several problems of interest. 

\begin{algorithm}[ht]
	\caption{\gsweep}
	\label{alg:gamma_sweep}
	\begin{algorithmic}
		\STATE {\bfseries Input:} utility $g$, cost $c$, algorithm $\mathcal{A}$, lower bound $L$, $\delta \in 
		(0,1)$
		\STATE $S_{-1} \gets \varnothing$, $T \gets \left\lceil \frac{1}{\delta} \ln \left( \frac{1}{\max \{ \delta, 
		L\}} 
		\right) \right \rceil$
		\FOR{$r=0$ {\bfseries to} $T$}
		\STATE $\gamma_r \gets (1 - \delta)^r$
		\STATE $S_r \gets \mathcal{A}(g, \gamma_r, c, k, \delta)$
		\ENDFOR
		\STATE {\bfseries Return} the set $R$ maximizing $g(R) - c(R)$ among $S_{-1}, S_0, \dotsc, S_T$.
	\end{algorithmic}
\end{algorithm}

In the following theorem, we assume that $\mathcal{A}(g, \gamma, c, 
k, 
\epsilon)$ is 
an algorithm which returns a set $S$ with $|S| \leq k$ and $ \Exp{ g(S) - c (S)} \geq \left( 1 - e^{-\gamma} - \epsilon 
\right) g(\opt) - c(\opt) $ when $g$ is $\gamma$-weakly submodular, and $L \leq \gamma$ is known 
(one may always use $L=0$).

\begin{mdframed}[nobreak=true]
\begin{theorem} \label{thm:guessing_gamma}
	\gsweep requires at most $O\left( \frac{1}{\delta} 
	 \log \frac{1}{\delta} \right)$ calls to $\mathcal{A}$ and returns a set $R$ with
	 \[
	 \Exp{ g(R) - c(R)} 
	 \geq \left( 1 - e^{-\gamma} - O(\delta) \right) g(\opt) - c(\opt) 
	 \enspace.
	 \]
\end{theorem}
\end{mdframed}
\begin{proof}
	We consider two cases. First, suppose that $\gamma < \delta$. Under this assumption, we have
	$$ 1 - e^{-\gamma} - \delta < 1 - e^{-\delta} - \delta \leq \delta - \delta = 0 \enspace,$$
	where the second inequality used the fact that $1 - e^{-x} \leq x$. Thus,
	$$ g(\varnothing) - c(\varnothing) \geq 0 \geq (1 - e^{-\gamma} - \delta) g(\opt) - c(\opt) 
	\enspace,$$
	where the first inequality follows from non-negativity of $g$, and the second inequality follows from 
	non-negativity of both $c$ and $g$. Because Algorithm~\ref{alg:gamma_sweep} sets $S^{(-1)} = 
	\varnothing$ and $R$ is chosen to be the 
	best solution,  
	$$g(R) - c(R) \geq g(\varnothing) - c(\varnothing) \geq \left( 1 - e^{-\gamma} - \delta \right) g(\opt) 
	- 
	c(\opt) 
	\enspace.
	$$
	For the second case, suppose that $\gamma \geq \delta$. Recall that $\gamma \geq L$ by 
	assumption, 
	and thus, ${\gamma \geq B \triangleq \max \{ \delta, L \}}$. Now, we need to show that $(1-\delta)^T 
	\leq B$. 
	This is equivalent to $\left( \frac{1}{1 - \delta} \right)^T \geq \frac{1}{B}$, and by taking $\ln$, 	this is 
	equivalent to $T \geq \frac{\ln \frac{1}{B}}{ \ln \frac{1}{1-\delta}}$. This is true since the inequality 
	$\delta \leq \ln \left( \frac{1}{1-\delta} \right)$, which holds for all $\delta \in (0,1)$, implies
	$$T = \left\lceil \frac{1}{\delta} \ln \frac{1}{B} \right\rceil 
	\geq \frac{1}{\delta} \ln \frac{1}{B} 
	\geq \frac{\ln \frac{1}{B}}{ \ln \frac{1}{1-\delta}} \enspace.$$
	Hence, we have proved that $(1-\delta)^T \leq B \leq \gamma$, which implies that there exists $t 
	\in \{0, 1, \dotsc, T\}$ such that $\gamma \geq \gamma^{(t)} \geq (1 - \delta) \gamma$. For notational 
	convenience, we write $\gammaapx = \gamma^{(t)}$. Because $g$ is $\gamma$-weakly submodular 
	and 
	$\gamma \geq \gammaapx$, $g$ is also $\gammaapx$-weakly submodular. Therefore, by 
	assumption, the algorithm $\mathcal{A}$ returns a set $S^{(t)}$ which satisfies
	$$\Exp{g(S^{(t)}) + \ell(S^{(t)})}  \geq \left( 1 - e^{-\gammaapx} - \delta \right) g(\opt) - c(\opt) 
	\enspace.$$
	From the convexity of $e^x$, we have $e^\delta \leq (1 - \delta)e^0 + \delta e^1 = 1 + (e-1)\delta$ for 
	all $\delta \in [0,1]$. Using 
	this inequality, and the fact that $\hat{\gamma} \geq (1 - \delta) \gamma$, we get
	\[
	1 - e^{-\gammaapx} 
	\geq 1 - e^{-(1 - \delta ) \gamma} 
	\geq 1 - e^{-\gamma} e^\delta 
	\geq 1 - e^{-\gamma}( 1 + (e-1) \delta)
	=    1 - e^{-\gamma} - \beta \delta 
	\enspace.
	\]
	We remark that $\beta \leq e-1 \approx 1.72$. Thus, by the non-negativity of $g$ and because the 
	output 
	set $R$ was chosen as the set with highest value,  
	\[\Exp{g(R) - c(R)} \geq \Exp{g(S^{(t)}) - c(S^{(t)})} \geq \left( 1 - e^{-\gamma} - (\beta \delta + \delta) 
	\right) 
	g(\opt) - c(\opt) \enspace. \qedhere\]
\end{proof}

\begin{figure*}[tbh!]
	\centering
	\hspace{0.21in} 
	\includegraphics[width=\textwidth]{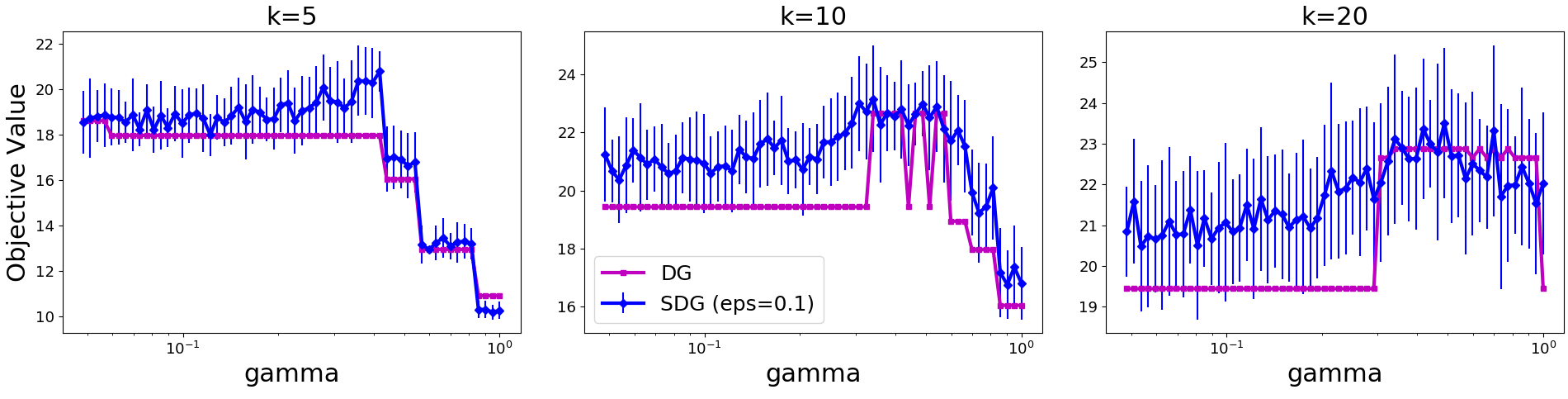}
	\caption{
		Results of the \gsweep with \dgreedy (DG) and \sdgreedy (SDG) as subroutines. For \sdgreedy , 
		mean values 
		with standard deviation bars are reported over 20 trials.
	}
	\label{fig:sweep_exp}
	\vspace{-0.1in}
\end{figure*}

In our experiments, we see that \sdgreedy combined with the \gsweep outperforms \dgreedy 
with \gsweep, especially for larger values of $k$.
Here, we provide some experimental evidence and explanation for why this may be occurring. 
Figure~\ref{fig:sweep_exp} shows the objective value of the sets $\{S_r\}_{r=0}^T$ produced by 
\sdgreedy and \dgreedy during the \gsweep for cardinality constraints $k=5$, $10$, and $20$.
Both subroutines return the highest objective value for similar ranges of $\gamma$.
However, the \sdgreedy subroutine appears to be better in two ways.
First, the average objective value is usually larger, meaning that an 
individual run of \sdgreedy is returning a higher quality set than \dgreedy.
This is likely due to the sub-sampling of the ground set, which might help avoiding the picking of a single ``bad 
element'', if one exists.
Second, the variation in \sdgreedy leads to a higher chance of producing a good solution.
For many values of $\gamma$, the \dgreedy subroutine returns a set of the same value; thus, the extra 
guesses of $\gamma$ are not particularly helpful.
On the other hand, the variation within the \sdgreedy subroutine means that these extra guesses are not 
wasted; in fact, they allow a higher chance of producing a set with good value.
Figure~\ref{fig:sweep_exp} also shows that the objective function throughout 
the sweep is fairly well-behaved, suggesting the possibility of early stopping heuristics. However, that 
is outside the scope of this paper.

\section{Hardness Result} \label{sec:hardness}

In this section, we give a hardness result which complements our algorithmic guarantees.
The hardness result shows that{\textemdash}in the case where $c=0${\textemdash}no algorithm 
making polynomially many queries to $g$ can achieve a better approximation ratio than $1-e^{-\gamma}$.
Although this was known in the case when $\gamma = 1$ (i.e., $g$ is submodular), the more 
general result for $\gamma < 1$ was unknown until this work.

\begin{mdframed}[nobreak=true]
\begin{theorem} \label{thm:weak_submodular_hardness}
	For every two constants $\eps > 0$ and $\gamma \in (0, 1]$, no polynomial time algorithm achieves 
	$(1 
	- e^{-\gamma} + \eps)$-approximation for the problem of maximizing a non-negative monotone 
	$\gamma$-weakly submodular function subject to a cardinality constraint in the value oracle model.
\end{theorem}
\end{mdframed}

As is usual in hardness proofs for submodular functions, the proof is based on constructing a family 
of $\gamma$-weakly submodular functions on which any deterministic algorithm will perform poorly in 
expectation, and then applying Yao's principle. It turns out that, instead of proving 
Theorem~\ref{thm:weak_submodular_hardness}, it is easier to 
prove a stronger theorem given below as Theorem~\ref{thm:weak_DR_hardness}. However, before we 
can present Theorem~\ref{thm:weak_DR_hardness}, we need the following definition (this definition is 
related to a notion called the \emph{DR-ratio} defined by~\citet{KSCT18} for functions over the integer 
lattice).



\begin{definition} \label{def:weak_DR}
A function $f\colon 2^\cN \to \bR$ is $\gamma$-weakly DR if for every two sets $A \subseteq B 
\subseteq \cN$ and element $u \in \cN \setminus B$ it holds that $f(u \mid A) \geq \gamma \cdot f(u 
\mid B)$.
\end{definition}
\begin{theorem} \label{thm:weak_DR_hardness}
For every two constants $\eps > 0$ and $\gamma \in (0, 1]$, no polynomial time algorithm achieves $(1 
- e^{-\gamma} + \eps)$-approximation for the problem of maximizing a non-negative monotone 
$\gamma$-weakly DR function subject to a cardinality constraint in the value oracle model.
\end{theorem}

The following observation shows that every instance of the problem considered by 
Theorem~\ref{thm:weak_DR_hardness} is also an instance of the problem considered by 
Theorem~\ref{thm:weak_submodular_hardness}, and therefore, Theorem~\ref{thm:weak_DR_hardness} 
indeed implies Theorem~\ref{thm:weak_submodular_hardness}.
\begin{observation}
A monotone $\gamma$-weakly DR set function $f\colon 2^\cN \to \nnR$ is also $\gamma$-weakly 
submodular.
\end{observation}
\begin{proof}
Consider arbitrary sets $A \subseteq B \subseteq \cN$, and let us denote the elements of the set $B 
\setminus A$ by $u_1, u_2, \dotsc, u_{|B \setminus A|}$ in a fixed arbitrary order. Then,
\[
	f(B \mid A)
	=
	\sum_{i = 1}^{|B \setminus A|} f(u_i \mid A \cup \{u_1, u_2, \dotsc, u_{i-1}\})
	\geq
	\gamma \cdot \sum_{i = 1}^{|B \setminus A|} f(u_i \mid A)
	\enspace.
	\qedhere
\]
\end{proof}

The following proposition is the main technical component used in the proof of 
Theorem~\ref{thm:weak_DR_hardness}. To facilitate the reading, we defer its proof to 
Section~\ref{sec:construction}.
\begin{proposition} \label{prop:construction}
For every value $\eps' \in (0, \nicefrac{1}{6})$, value $\gamma \in (0, 1]$ and integer $k \geq 1/\eps'$, there exists a ground set $\cN$ of size $\lceil 3k / \eps' \rceil$ and a set function $f_T\colon 
2^\cN \to \nnR$ for every set $T \subseteq \cN$ of size at most $k$ such that
\begin{compactenum}[\bf(P1)]
	\item $f_T$ is non-negtive monotone and $\gamma$-weakly DR. \label{property:named_properties}
	\item $f_T(S) \leq 1$ for every set $S \subseteq \cN$, and the inequality holds as an equality for $S = 
	T$ when the size of $T$ is exactly $k$. \label{property:upper_bound}
	\item $f_\varnothing(S) \leq 1 - e^{-\gamma} + 12\eps'$ for every set $S$ of size at most $k$. 
	\label{property:empty_behavior}
	\item $f_T(S) = f_\varnothing(S)$ when $|S| \geq 3k - g$ or $|S \cap T| \leq g$, where $g = \lceil \eps' 
	k + 3k^2 / |\cN|\rceil$. \label{property:different_T_relation}
\end{compactenum}
\end{proposition}

At this point, let us consider some $\gamma$ value and set $\eps' = \eps/20$. Note that Theorem~\ref{thm:weak_DR_hardness} is trivial for $\eps > 1$, and thus, we may assume $\eps' \in (0, 
\nicefrac{1}{6})$, which implies that there exists a large enough integer $k$ for which $\gamma$, $\eps'$ and $k$ obey all 
the requirements of Proposition~\ref{prop:construction}. From this point on we consider the ground 
set $\cN$ and the functions $f_T$ whose existence is guaranteed by 
Proposition~\ref{prop:construction} for these values of $\gamma$, $\eps'$ and $k$. Let $\tilde{T}$ be 
a random subset of $\cN$ of size $k$ (such subsets exist because $|\cN| > k$). Intuitively, in the rest 
of this section we prove Theorem~\ref{thm:weak_DR_hardness} by showing that the problem 
$\max\{f_{\tilde{T}}(S) \mid S \subseteq \cN, |S| \leq k\}$ is hard in expectation for every algorithm.

Property~\property{property:upper_bound} of Proposition~\ref{prop:construction} shows that the 
optimal solution for the problem  $\max\{f_{\tilde{T}}(S) \mid S \subseteq \cN, |S| \leq k\}$ is $\tilde{T}$. 
Thus, an algorithm expecting to get a good approximation ratio for this problem should extract 
information about the random set $\tilde{T}$. The question is on what sets should the algorithm 
evaluate $f_{\tilde{T}}$ to get such information. Property~\property{property:different_T_relation} of 
the proposition shows that the algorithm cannot get much information about $\tilde{T}$ when querying 
$f_{\tilde{T}}$ on a set $S$ that is either too large or has a too small intersection with $\tilde{T}$. Thus, 
the only way in which the algorithm can get a significant amount of information about $\tilde{T}$ is by 
evaluating $f_{\tilde{T}}$ on a set $S$ that is small and not too likely to have a small intersection with 
$\tilde{T}$. Lemma~\ref{lem:probability} shows that such sets do not exist. However, before we can 
prove Lemma~\ref{lem:probability}, we need the following known lemma.
\begin{lemma}[Proved by~\citet{S13} based on results of~\cite{C79} and~\cite{H63}] \label{lem:hypergeometric}
Consider a population of $N$ balls, out of which $M$ are white. Given a hypergeometric variable $X$ 
measuring the number of white balls obtained by drawing uniformly at random $n$ balls from this 
population, it holds for every $t \geq 0$ that $\Pr[X \geq nM/N + tn] \leq e^{-2t^2n}$.
\end{lemma}
\begin{lemma} \label{lem:probability}
For every set $S \subseteq \cN$ whose size is less than $3k - g$, $\Pr[|S \cap \tilde{T}| \leq g] \geq 1 - 
e^{-\Omega(\eps^3|\cN|)}$.
\end{lemma}
\begin{proof}
%
The distribution of $|S \cap \tilde{T}|$ is hypergeometric. More specifically, it is equivalent to drawing 
$k$ balls from a population of $|\cN|$ balls, of which only $|S|$ are white. Thus, by 
Lemma~\ref{lem:hypergeometric}, for every $t \geq 0$ we have
\[
	\Pr[|S \cap \tilde{T}| \geq k|S|/|\cN| + tk]
	\leq
	e^{-2t^2k}
	\enspace.
\]
Setting $t = \eps'$ and observing that $|S| \leq 3k - g \leq 3k$, the last inequality yields
\[
	\Pr[|S \cap \tilde{T}| \geq 3k^2/|\cN| + \eps' k]
	\leq
	\exp\left(-2(\eps')^2k\right)
	=
	\exp\left(-\frac{\eps^2k}{200}\right)
	\enspace.
\]
The lemma now follows since $g \geq 3k^2/|\cN| + \eps' k$, and (by the definition of $\cN$)
\[
	k
	\geq
	\frac{\eps'(|\cN| - 1)}{3}
	=
	\frac{\eps(|\cN| - 1)}{60}
	=
	\Omega(\eps|\cN|)
	\enspace.
	\qedhere
\]
\end{proof}
\begin{corollary} \label{cor:equality_probability}
For every set $S \subseteq \cN$, $\Pr[f_\varnothing(S) = f_{\tilde{T}}(S)] \geq 1 - 
e^{-\Omega(\eps^3|\cN|)}$.
\end{corollary}
\begin{proof}
If $|S| \geq 3k - g$, then the corollary follows from Property~\property{property:different_T_relation} of 
Proposition~\ref{prop:construction}. Otherwise, it follows by combining this property with 
Lemma~\ref{lem:probability}.
\end{proof}

Using the above results, we are now ready to prove an hardness result for deterministic algorithms.
\begin{lemma} \label{lem:hardness_deterministinc}
Consider an arbitrary deterministic algorithm $ALG$ for the problem $\max\{f(S) \mid S \subseteq \cN, 
|S| \leq k\}$ whose time complexity is bounded by some polynomial function $C(|\cN|)$. Then, there is 
a large enough value $k$ that depends only on $C(\cdot)$ and $\eps$ such that, given the random 
instance $\max\{f_{\tilde{T}}(S) \mid S \subseteq \cN, |S| \leq k\}$ of the above problem,  the expected 
value of the output set of $ALG$ is no better than $1 - e^{-\gamma} + \eps$.
\end{lemma}
\begin{proof}
Let $S_1, S_2, \dotsc, S_\ell$ be the sets on which $ALG$ evaluate $f_\varnothing$ when it is given the 
instance $\max\{f_\varnothing(S) \mid S \subseteq \cN, |S| \leq k\}$, and $S_{\ell + 1}$ be its output set 
given this instance. Let $\cE$ be the event that $f
_\varnothing(S_i) = f_{\tilde{T}}(S_i)$ for every $1 \leq i \leq \ell + 1$. By combining 
Corollary~\ref{cor:equality_probability} with the union bound, we get that
\[
	\Pr[\cE]
	\geq
	1 - (\ell + 1) \cdot e^{-\Omega(\eps^3|\cN|)}
	\geq
	1 - [C(|\cN|) + 1] \cdot e^{-\Omega(\eps^3|\cN|)})
	\enspace,
\]
where the second inequality holds since the time complexity of an algorithm upper bounds the number 
of sets on which it may evaluate $f_\varnothing$. Since $C(|\cN|)$ is a polynomial function, by making 
$k$ large enough, we can make $\cN$ large enough to guarantee that $C(|\cN|) \cdot 
e^{-\Omega(\eps^3|\cN|)} \leq \eps/20$, and thus, $\Pr[\cE] \geq 1 - \eps/20$.

When the event $\cE$ happens, the values that $ALG$ gets when evaluating $f_{\tilde{T}}$ on the sets 
$S_1, S_2, \dotsc,\allowbreak S_\ell$ is equal to the values that it would have got if the objective 
function was $f_\varnothing$. Thus, in this case $ALG$ follows the same execution path as when it 
gets $f_\varnothing$, and outputs $S_{\ell + 1}$ whose value is
\[
	f_{\tilde{T}}(S_{\ell + 1})
	=
	f_{\varnothing}(S_{\ell + 1})
	\leq
	1 - e^{-\gamma} + 12\eps'
	=
	1 - e^{-\gamma} + 3\eps/5
	\enspace,
\]
where the inequality holds by Property~\property{property:empty_behavior} of 
Proposition~\ref{prop:construction} since the output set $S_{\ell + 1}$ must be a feasible set, and thus, 
of size at most $k$. When the event $\cE$ does not happen, we can still upper bound the value of the 
output set of $ALG$ by $1$ using Property~\property{property:upper_bound} of the same proposition. 
Thus, if we denote by $R$ the output set of $ALG$, then, by the law of iterated expectations,
\begin{align*}
	\bE[f_{\tilde{T}}(R)]
	={} &
	\Pr[\cE] \cdot \bE[f_{\tilde{T}}(S_{\ell + 1}) \mid \cE] + \Pr[\neg \cE] \cdot \bE[f_{\tilde{T}}(R) \mid \neg 
	\cE]\\
	\leq{} &
	1 \cdot (1 - e^{-\gamma} + 3\eps/5) + (\eps/20) \cdot 1
	=
	1 - e^{-\gamma} + 13\eps/20
	\leq
	1 - e^{-\gamma} + \eps
	\enspace.
	\qedhere
\end{align*}
\end{proof}

Lemma~\ref{lem:hardness_deterministinc} shows that there is a single distribution of instances which is 
hard for every deterministic algorithm whose time complexity is bounded by a polynomial function 
$C(|\cN|)$. Since a randomized algorithm whose time complexity is bounded by $C(|\cN|)$ is a 
distribution over deterministic algorithms of this kind, by Yao's principle, 
Lemma~\ref{lem:hardness_deterministinc} yields the next corollary.

\begin{corollary} \label{cor:hardness_randomized}
Consider an arbitrary algorithm $ALG$ for the problem $\max\{f(S) \mid S \subseteq \cN, |S| \leq k\}$ 
whose time complexity is bounded by some polynomial function $C(|\cN|)$. Then, there is a large 
enough value $k$ that depends only on $C(\cdot)$ such that, for some set $T \subseteq \cN$ of size 
$k$, given the instance $\max\{f_T(S) \mid S \subseteq \cN, |S| \leq k\}$ of the above problem,  the 
expected value of the output set of $ALG$ is no better than $1 - e^{-\gamma} + \eps$.
\end{corollary}

Theorem~\ref{thm:weak_DR_hardness} now follows from Corollary~\ref{cor:hardness_randomized} 
because Property~\property{property:upper_bound} shows that the optimal solution for the instance 
$\max\{f_T(S) \mid S \subseteq \cN, |S| \leq k\}$ mentioned by this corollary has a value of $1$, and 
Property~\property{property:named_properties} of the same proposition shows that this instance is an 
instance of the problem of maximizing a non-negative monotone $\gamma$-weakly-DR function 
subject to a cardinality constraint.

\subsection{Proof of Proposition~\ref{prop:construction}} \label{sec:construction}

In this section we prove Proposition~\ref{prop:construction}. We begin the proof by defining the 
function $f_T$ whose existence is guaranteed by the proposition. To define $f_T$, we first need to 
define the following four helper functions. Note that in $f_{T, 2}$ we use the notation $[x]^+$ to 
denote the maximum between $x$ and $0$.
\begin{center}\begin{tabular}{ll}
\textbullet\hspace{2mm} $\displaystyle
	t_T(S)
	\triangleq
	|S \setminus T| + \min\{g, |S \cap T|\}
$
&
\textbullet\hspace{2mm} $\displaystyle
	f_{T,2}(S) \triangleq 1 - \frac{\min\{[t_T(S) - k]^+, k - g\}}{k - g}
$\\[5mm]
\textbullet\hspace{2mm} $\displaystyle
	f_{T,1}(S) \triangleq \left(1 - \frac{\gamma}{k - g}\right)^{\min\{t_T(S), k\}}
$
&
\textbullet\hspace{2mm} $\displaystyle
	f_{T,3}(S) \triangleq 1 - \frac{\min\{|S| - t_T(S), k - g\}}{k - g}
	\enspace.
$
\end{tabular}\end{center}

Using these helper functions, we can now define $f_T$ for every set $S \subseteq \cN$ by
\[
	f_T(S)
	\triangleq
	1 - f_{T,1}(S) \cdot f_{T,2}(S) \cdot f_{T,3}(S)
	\enspace.
\]

In the rest of this section we show that the function $f_T$ constructed this way obeys all the 
properties guaranteed by Proposition~\ref{prop:construction}. We begin with the following technical 
observation that comes handy in some of our proofs.
\begin{observation} \label{obs:g_smaller_k}
$g \leq 2\eps' k + 1 \leq \min\{k - 2, 3\eps' k\}$.
\end{observation}
\begin{proof}
The second inequality of the observation follows immediately from the assumptions of 
Proposition~\ref{prop:construction} regarding $k$ and $\eps'$ (\ie, the assumptions that $k \geq 1 / \eps'$ and $\eps' \in (0, \nicefrac{1}{6})$). Thus, we only need to prove the first inequality. 
Since $|\cN| \geq 3k / \eps'$,
\[
	g
	=
	\left \lceil \eps' k + \frac{3k^2}{|\cN|} \right\rceil
	\leq
	\eps' k + \frac{3k^2}{3k / \eps'} + 1
	=
	2\eps' k + 1
	\enspace.
	\qedhere
\]
\end{proof}

The next three lemmata prove together Property~\property{property:named_properties} of 
Proposition~\ref{prop:construction}.

\begin{lemma} \label{lem:non_negative}
The outputs of the functions $f_{T,1}$, $f_{T,2}$ and $f_{T,3}$ are always within the range $[0, 1]$, 
and thus, $f_T$ is non-negative.
\end{lemma}
\begin{proof}
We prove the lemma for every one of the functions $f_{T,1}$, $f_{T,2}$ and $f_{T,3}$ separately.
\begin{itemize}
	\item Let $b = 1 - \gamma / (k - g)$. One can observe that $f_{T, 1}$ is defined as $b$ to the power 
	of $\min\{t_T(S), k\}$. Thus, to show that the value of $f_{T, 1}$ always belongs to the range $[0, 1]$, 
	it suffices to prove that $b \in (0, 1]$ and $\min\{t_T(S), k\}$ is non-negative. The first of these claims 
	holds since $\gamma \in (0, 1]$ by assumption and $k - g \geq 2$ by 
	Observation~\ref{obs:g_smaller_k}, and the second claim can be verified by looking at the definition 
	of $t_T(S)$ and noting that $g$ must be positive.
	\item Since $k - g \geq 0$ by Observation~\ref{obs:g_smaller_k}, $\min\{[t_T(S) - k]^+, k - g\} \in [0, k 
	- g]$. Plugging this result into the definition of $f_{T, 2}$ yields that the value of $f_{T, 2}$ always 
	belongs to $[0, 1]$.
	\item Note that the definition of $t_T(S)$ implies $t_T(S) \leq |S|$. Together with the inequality $k - g 
	\geq 0$, which holds by Observation~\ref{obs:g_smaller_k}, this guarantees $\min\{|S| - t_T(S), k - 
	g\} \in [0, k - g]$. Plugging this result into the definition of $f_{T, 3}$ yields that the value of $f_{T, 
	3}$ always belongs to $[0, 1]$. \qedhere
\end{itemize}
\end{proof}

We say that a set function $h\colon 2^\cN \to \bR$ is monotonically decreasing if $f(A) \geq f(B)$ for 
every two sets $A \subseteq B \subseteq \cN$.

\begin{lemma} \label{lem:monotonicity}
The functions $f_T$ and $|S| - t_T(S)$ are monotone and the functions $f_{T,1}$, $f_{T,2}$ and 
$f_{T,3}$ are monotonically decreasing.
\end{lemma}
\begin{proof}
It immediately follows from the definition of $t_T(S)$ that it is monotone. Additionally, $|S| - t_T(S)$ is 
a monotone function since it is equal to
\[
	|S| - t_T(S)
	=
	|S \cap T| - \min\{g, |S \cap T|\}
	=
	[|S \cap T| - g]^+
	\enspace.
\]
Plugging these observations into the definitions of $f_{T,1}$, $f_{T,2}$ and $f_{T,3}$ yields that these 
three functions are all monotonically decreasing. Since these three functions are also non-negative by 
Lemma~\ref{lem:non_negative}, this implies that $f_{T,1}(S) \cdot f_{T,2}(S) \cdot f_{T, 3}(S)$ is also a 
monotonically decreasing function, and thus, $f_T$ is a monotone function since it is equal to $1$ 
minus this product.
\end{proof}

\begin{lemma}
$f_T$ is $\gamma$-weakly-DR.
\end{lemma}
\begin{proof}
Consider arbitrary sets $A \subseteq B \subseteq \cN$, and fix an element $u \in \cN \setminus B$. We 
need to show that $f_T(u \mid A) \geq \gamma \cdot f_T(u \mid B)$, which we do by considering three 
cases.

The first case is when $t_T(A \cup \{u\}) = t_T(A) + 1$ and $t_T(B \cup \{u\}) = t_T(B) + 1$. Note that for every set 
$S$ for which $t_T(S \cup \{u\}) = t_T(S) + 1$ and $t_T(S) < k$ we have
\begin{align} \label{eq:increasing_smaller_k}
	f_T(u \mid S)
	={} &
	f_{T,1}(S) \cdot f_{T,2}(S) \cdot f_{T,3}(S) - f_{T,1}(S \cup \{u\}) \cdot f_{T,2}(S \cup \{u\}) \cdot f_{T,3}(S \cup \{u\}) \nonumber\\ 
	={} &
	f_{T,1}(S) \cdot f_{T,3}(S) - \left(1 - \frac{\gamma}{k - g}\right) \cdot f_{T,1}(S) \cdot f_{T,3}(S) \\
	\nonumber
	={} &
	\frac{\gamma}{ k - g} \cdot f_{T,1}(S) \cdot f_{T,3}(S)
	\enspace,
\end{align}
and for every set $S$ for which $t_T(S + u) = t_T(S) + 1$ and $t_T(S) \geq k$ we have
\begin{align} \label{eq:increasing_larger_k}
	f_T(u \mid {}&S)
	=
	f_{T,1}(S) \cdot f_{T,2}(S) \cdot f_{T,3}(S) - f_{T,1}(S \cup \{u\}) \cdot f_{T,2}(S \cup \{u\}) \cdot f_{T,3}(S \cup \{u\})\nonumber\\ 
	={} &
	f_{T,1}(S) \cdot f_{T,2}(S) \cdot f_{T,3}(S) \nonumber \\ 
	&\quad - f_{T,1}(S) \cdot \left[f_{T,2}(S) - \frac{\min\{[(k - g) - 
	(t_T(S) - k)]^+, 1\}}{k - g} \right] \cdot f_{T,3}(S)\\ \nonumber
	={} &
	\frac{\min\{[(k - g) - (t_T(S) - k)]^+, 1\}}{k - g} \cdot f_{T,1}(S) \cdot f_{T,3}(S)
	\leq
	\frac{1}{k - g} \cdot f_{T,1}(S) \cdot f_{T,3}(S)
	\enspace,
\end{align}
where the last inequality holds since $f_{T,1}$ and $f_{T,3}$ are non-negative by 
Lemma~\ref{lem:non_negative}. Since $f_1$ and $f_3$ are monotonically decreasing functions (by 
Lemma~\ref{lem:monotonicity}), the above inequalities show $f_T(u \mid A) \geq \gamma \cdot f_T(u 
\mid B)$ whenever $t_T(A) < k$---if $t_T(B) < k$, then the inequalities in fact show $f_T(u \mid A) \geq 
f_T(u \mid B)$, but this implies $f_T(u \mid A) \geq \gamma \cdot f_T(u \mid B)$ because $f_T$ is 
monotone and $\gamma \in (0, 1]$. It remains to consider the option $t(A) \geq k$. Note that when this 
happens, we also have $t_T(B) \geq k$ because $t_T(S)$ is a monotone function. Thus, $f_T(u \mid A) \geq 
f_T(u \mid B)$ because $f_{T,1}$, $f_{T,3}$ and $\min\{[(k - g) - (t_T(S) - k)]^+, 1\}$ are all non-negative 
monotonically decreasing functions, and like in the above, this implies $f_T(u \mid A) \geq \gamma \cdot 
f_T(u \mid B)$.

The second case we consider is when $t_T(A \cup \{u\}) = t_T(A)$. Note that in this case we also have $t_T(B 
\cup \{u\}) = t_T(B)$ because the equality $t_T(A \cup \{u\}) = t_T(A)$ implies $g = \min\{|A \cap T|, g\} \leq \min\{|B \cap T|, g\} 
\leq g$, which implies in its turn $\min\{|B \cap T|, g\} = g$. For every set $S$ for which $t_T(S \cup \{u\}) = t_T(S)$ and $u \not \in S$ we have
\begin{align} \label{eq:not_increasing}
	f_T(u \mid{}& S)
	=
	f_{T,1}(S) \cdot f_{T,2}(S) \cdot f_{T,3}(S) - f_{T,1}(S \cup \{u\}) \cdot f_{T,2}(S \cup \{u\}) \cdot f_{T,3}(S \cup \{u\}) \nonumber\\ 
	={} &
	f_{T,1}(S) \cdot f_{T,2}(S) \cdot f_{T,3}(S) \nonumber \\  
	& \quad - f_{T,1}(S) \cdot f_{T,2}(S) \cdot \left[f_{T,3}(S) - 
	\frac{\min\{[(k - g) - (|S| - t_T(S))]^+, 1\}}{k -g}\right]\\ \nonumber
	={} &
	f_{T,1}(S) \cdot f_{T,2}(S) \cdot \frac{\min\{[(k - g) - (|S| - t_T(S))]^+, 1\}}{k -g}
	\enspace.
\end{align}
Recall now that $f_{T,1}$ and $f_{T,3}$ are non-negative and monotonically decreasing functions by 
Lemmata~\ref{lem:non_negative} and~\ref{lem:monotonicity}. We additionally observe that the function 
$\min\{[(k - g) - (|S| - t_T(S)]^+, 1\}$ also has these properties because Lemma~\ref{lem:monotonicity} 
shows that $|S| - t(S)$ is monotone. Combining these facts, we get that the expression we obtained 
for $f(u \mid S)$ is a monotonically decreasing function of $S$. Thus, $f(u \mid A) \geq f(u \mid B)$, 
which implies $f(u \mid A) \geq \gamma \cdot f(u \mid B)$.

The last case we need to consider is the case that $t_T(A \cup \{u\}) = t_T(A) + 1$ and $t_T(B \cup \{u\}) = t_T(B)$. 
The fact that $t_T(B \cup \{u\}) = t_T(B)$ implies that $u \in T$, and therefore, the fact that $t(A \cup \{u\}) = t(A) + 
1$ implies that $|A \cap T| < g$ and $t_T(A) = |A|$, which induces in its turn $f_{T,3}(A) = 1$. There are 
now a few sub-cases to consider. If $t_T(A) < k$, then
\begin{align*}
f_T(u \mid A)
&=
\frac{\gamma}{k - g} \cdot f_{T,1}(A) \\
&\geq
\gamma \cdot f_{T,1}(B) \cdot f_{T,2}(B) \cdot \frac{\min\{[(k - g) - (|B| - t_T(B))]^+, 1\}}{k - g}
=
\gamma \cdot f_T(u \mid B)
\enspace,
\end{align*}
where the first equality holds by Equation~\eqref{eq:increasing_smaller_k}, the last equality holds 
by Equation~\eqref{eq:not_increasing}, and the 
inequality holds since $\min\{[(k - g) - (|B| - t_T(B))]^+, 1\} \in [0, 1]$, $f_{T, 1}(A) \geq f_{T, 1}(B) \geq 
0$ by Lemmata~\ref{lem:non_negative} and~\ref{lem:monotonicity} and $f_{T,2}(B) \in [0, 1]$ by 
Lemma~\ref{lem:non_negative} . The second sub-case we need to consider is when $k \leq t(A) \leq 2k 
- g - 1$. In this case
\begin{align*}
	f_T(u \mid A)
	={} &
	\frac{\min\{[(k - g) - (t_T(A) - k)]^+, 1\}}{k - g} \cdot f_{T,1}(A)
	=
	\frac{1}{k - g} \cdot f_{T,1}(A)\\
	\geq{} &
	\frac{\gamma}{k - g} \cdot f_{T,1}(B)\\
	\geq{} &
	\frac{\gamma \cdot \min\{[(k - g) - (|B| - t_T(B))]^+, 1\}}{k - g} \cdot f_{T,1}(B) \cdot f_{T,2}(B)
	=
	\gamma \cdot f_T(u \mid B)
	\enspace,
\end{align*}
where the first equality holds by Equation~\eqref{eq:increasing_larger_k} and the last equality 
holds by Equation~\eqref{eq:not_increasing}. The first inequality holds since $\gamma \in (0, 1]$ and 
$f_{T, 1}$ is non-negative and monotonicity decreasing, and the second inequality holds since 
$\min\{[(k - g) - (|B| - t_T(B))]^+, 1\}$ and $f_{T, 2}(B)$ are both values in the range $[0, 1]$ and 
$\gamma \cdot f_{T, 1}(B) / (k - g)$ is non-negative. The final sub-case we consider is the case in 
which $t_T(A) \geq 2k - g - 1$. Since $|T \cap A| < g$ (but $|T \cap B| \geq g$), in this 
sub-case we must have $t_T(B) \geq 2k - g$, which implies $f_{T,2}(B) = 0$, and thus,
\[
	f_T(u \mid B)
	=
	f_{T,1}(B) \cdot f_{T,2}(B) \cdot \frac{\min\{[(k - g) - (|B| - t_T(B))]^+, 1\}}{k -g}
	=
	0
	\leq
	f_T(u \mid A)
	\enspace,
\]
where the equality holds by Equation~\eqref{eq:not_increasing}, and the inequality follows from the 
monotonicity of $f_T$.
\end{proof}

This completes the proof of Property~\property{property:named_properties} of 
Proposition~\ref{prop:construction}. The next lemma proves 
Property~\property{property:upper_bound} of this proposition.
\begin{lemma}
$f_T(S) \leq 1$ for every set $S \subseteq \cN$, and the inequality holds as an equality for $S = T$ 
when the size of $T$ is exactly $k$.
\end{lemma}
\begin{proof}
Since $f_{T, 1}$, $f_{T, 2}$ and $f_{T, 3}$ all output only values within the range $[0, 1]$ by 
Lemma~\ref{lem:non_negative}, $f_T(S) = 1 - f_{T, 1}(S) \cdot f_{T, 2}(S) \cdot f_{T, 3}(S) \leq 1$. 
Additionally, since $g \leq k$ by Observation~\ref{obs:g_smaller_k}, $t_T(T) = g$ when $|T| = k$. 
Hence, for such $T$,
\[
	f_{T,3}(T)
	=
	1 - \frac{\min\{k - g, k - g\}}{k - g}
	=
	0
	\enspace,
\]
which implies, $f_T(T) = 1 - f_{T,1}(T) \cdot f_{T, 2}(T) \cdot f_{T,3}(T) = 1$.
\end{proof}

The next lemma proves Property~\property{property:empty_behavior} of 
Proposition~\ref{prop:construction}.
\begin{lemma}
$f_\varnothing(S) \leq 1 - e^{-\gamma} + 8\eps'$ for every set $S$ obeying $|S| \leq k$.
\end{lemma}
\begin{proof}
Consider an arbitrary set $S$ obeying $|S| \leq k$. Note that for such a set we have $t_\varnothing(S) 
= |S| \leq k$. Hence,
\[
	f_{\varnothing, 2}(S)
	=
	f_{\varnothing, 3}(S)
	=
	1 - \frac{\min\{0, k - g\}}{k - g}
	=
	1
	\enspace.
\]
Therefore,
\begin{align} 
	\nonumber
	f_\varnothing(S)
	&=
	1 - f_{\varnothing, 1}(S) \cdot f_{\varnothing, 2}(S) \cdot f_{\varnothing, 3}(S)
	=
	1 - f_{\varnothing, 1}(S) \\
	&=
	1 - \left(1 - \frac{\gamma}{k - g}\right)^{|S|}
	\leq
	1 - \left(1 - \frac{\gamma}{k - g}\right)^{k}
	\enspace. \label{eq:f_S_first}
\end{align}
To prove the lemma, we need to upper bound the rightmost side of the last inequality. Towards this 
goal, observe that
\begin{align} 
	\nonumber
	\left(1 - \frac{\gamma}{k - g}\right)^{k}
	&\geq
	\left(1 - \frac{\gamma}{k - 3\eps' k}\right)^{k} \\
	&=
	\left(1 - \frac{\gamma}{k - 3\eps' k}\right)^{k - 3\eps' k} \cdot \left(1 - \frac{\gamma}{k - 3\eps' 
	k}\right)^{3\eps' k}
	\enspace, \label{eq:f_S_second} 
\end{align}
where the first inequality holds since $g \leq 3\eps' k$ by Observation~\ref{obs:g_smaller_k}. Let us 
now lower bound the two factors in the product on the rightmost side. First,
\[
	\left(1 - \frac{\gamma}{k - 3\eps' k}\right)^{k - 3\eps' k}
	\geq
	e^{-\gamma}\left(1 - \frac{\gamma^2}{k - 3\eps' k}\right)
	\geq
	e^{-\gamma}\left(1 - 2\eps'\right)
	\enspace,
\]
where the first inequality holds since the assumptions of Proposition~\ref{prop:construction} imply $k 
- 3\eps'k \geq k/2 \geq 1$, and the second inequality holds since these assumptions include $k \geq 
1/\eps'$ and $\gamma \in (0, 1]$. Additionally,
\[
	\left(1 - \frac{\gamma}{k - 3\eps' k}\right)^{3\eps' k}
	\geq
	\left(1 - \frac{2}{k}\right)^{3\eps' k}
	\geq
	1 - \frac{2}{k} \cdot (3\eps' k)
	=
	1 - 6\eps'
	\enspace,
\]
where the first inequality holds again since $\gamma \in (0, 1]$ and $k - 3\eps' k \geq k/2$. Plugging 
the last two lower bounds into Inequality~\eqref{eq:f_S_second} and combining with 
Inequality~\eqref{eq:f_S_first}, we get
\[
	f_\varnothing(S)
	\leq
	1 - e^{-\gamma}(1 - 2\eps') \cdot \left(1 - 6\eps'\right)
	\leq
	1 - e^{-\gamma} \cdot \left(1 - 8\eps'\right)
	\leq
	1 - e^{-\gamma} + 8\eps'
	\enspace.
	\qedhere
\]
\end{proof}

To complete the proof of Proposition~\ref{prop:construction}, it remains to prove 
Property~\property{property:different_T_relation}, which is done by the next two observations.
\begin{observation}
If $|S \cap T| \leq g$, then $f_T(S) = f_\varnothing(S)$.
\end{observation}
\begin{proof}
The only place in the definition of $f_T(S)$ in which the set $T$ is used is in the definition of $t_T(S)$. 
Thus, $f_T(S) = f_{T'}(S)$ whenever $t_T(S) = t_{T'}(S)$. In particular, one can note that the condition 
$|S \cap T| \leq g$ implies $t_T(S) = |S| = t_\varnothing(S)$, and thus, $f_T$ and $f_\varnothing$ must 
agree on the set $S$.
\end{proof}

\begin{observation}
The equality $f_T(S) = 1$ holds for every set $S$ of size at least $3k - g$ and set $T$ of size at most 
$k$.
\end{observation}
\begin{proof}
Note that $t_T(S) \geq |S \setminus T| \geq |S| - |T| \geq (3k - g) - k = 2k - g$. Thus,
\[
	f_{T, 2}(S)
	=
	1 - \frac{\min\{[t_T(S) - k]^+, k - g\}}{k - g}
	=
	1 - \frac{k - g}{k - g}
	=
	0
	\enspace,
\]
which implies $f_T(S) = 1 - f_{T,1}(S) \cdot f_{T,2}(S) \cdot f_{T,3}(S) = 1$.
\end{proof}

\section{Experiments} \label{sec:experiments}
To demonstrate the effectiveness of our proposed algorithms, we run experiments on two applications: 
Bayesian $A$-optimal design with costs and directed vertex cover with costs. The code was written using 
the \texttt{Julia} programming language, version 1.0.2. Experiments were run on a 2015 MacBook Pro 
with 3.1 GHz Intel Core i7 and 8 GB DDR3 SDRAM and the timing was reported using the 
\texttt{@timed} feature in \texttt{Julia}. 
The source code is available on a public GitHub repository.%
\footnote{\url{https://github.com/crharshaw/submodular-minus-linear}}

\subsection{Bayesian \texorpdfstring{$A$}{A}-Optimal Design} \label{sec:a_optimal}

We first describe the problem of Bayesian $A$-Optimal design. Suppose that $\theta \in \reals^d$ is an 
unknown 
parameter vector that we wish to estimate from noisy linear measurements using least squares 
regression. Our goal is to choose a set $S$ of linear measurements (the so-called experiments) which 
have low cost and also maximally reduce the variance of our estimate $\hat{\theta}$. More precisely, let 
$x_1, x_2, \dots, x_n \in \reals^d$ be a fixed set of measurement vectors, and let $X = [ x_1, x_2, \dots, 
x_n]$ be the corresponding $d \times n$ matrix. Given a set of measurement vectors 
$S \subseteq [n]$, we may run the experiments and obtain the noisy linear observations, 
\[
y_S = X_S^T \theta + \zeta_S \enspace,
\]
where $\zeta_S$ is normal i.i.d. noise, i.e., $\zeta_1, \dotsc, \zeta_n \sim N(0, \sigma^2)$. 
We estimate $\theta$ using the least squares estimator
$ \hat{\theta} = (X_S X_S^T)^{-1}X_S^T y_S $. Assuming a normal Bayesian prior distribution on the 
unknown parameter, $\theta \sim N(0, \Sigma)$, the sum of the variance of the coefficients given the 
measurement set $S$ is 
$r(S) = \trinv{ \Sigma^{-1} + \frac{1}{\sigma^2} X_S X_S^T } $.
We define $g(S) = r(\varnothing) - r(S)$ to be the \emph{reduction in variance} produced by experiment 
set~$S$.

Suppose that each experiment $x_i$ has an associated non-negative cost $c_i$. 
In this application, we seek to maximize the ``revenue'' of the experiment, 
\[
g(S) - c(S) = \tr{\Sigma} -  \trinv{ \Sigma^{-1} + \frac{1}{\sigma^2} X_S X_S^T } \mspace{-13mu} - c(S) 
\enspace,
\]
which trades off the utility of the experiments (i.e., the variance 
reduction in the estimator) and their overall cost.

\citet{Bian2017} showed that $g$ is $\gamma$-weakly submodular, providing a lower bound for 
$\gamma$ in the case where $\Sigma = \beta I$. However, their bound relies rather unfavorably on the 
spectral norm of $X$, and does not extend to general $\Sigma$. 
\citet{chamon2017} showed that $g$ satisfies the stronger condition of $\gamma$-weak DR 
 (Definition~\ref{def:weak_DR}), 
but their bound on the submodularity ratio $\gamma$ depends on the cardinality of the sets.
We give a tighter bound here which 
relies on the Matrix Inversion Lemma (also known as Woodbury Matrix 
Identity and Sherman-Morrison-Woodbury Formula).

\begin{lemma}[Woodbury] \label{lemma:woodbury}
	For matrices $A$, $C$, $U$, and $V$ of the appropriate sizes,	
	$$(A + UCV)^{-1} = A^{-1} - A^{-1} U( C^{-1} + V A^{-1} U)^{-1} V A^{-1} $$
	In particular, for a matrix $A$, a vector $x$, and a number $\alpha$, we have that
	$$ \left( A + \frac{1}{\alpha} x x^T \right)^{-1} = A^{-1} - \frac{A^{-1} x x^T A^{-1}}{ \alpha + x^T A^{-1} 
		x} \enspace.$$
\end{lemma}

\begin{claim} \label{claim:a_opt_gamma_lb}
	$g$ is a non-negative, monotone and $\gamma$-weakly submodular function with 
	$$\gamma \geq \left(1 + \frac{s^2}{\sigma^2} \lambda_{\max}(\Sigma) \right)^{-1} \enspace,$$
	where $s = \max_{i \in [n]} \| x_i \|_2$.
\end{claim}
\begin{proof}
	Recall that 
	\[
	g(S) = \tr{\Sigma} -  \trinv{ \Sigma^{-1} + \frac{1}{\sigma^2} X_S X_S^T } \enspace.
	\]
	Let $A, B \subseteq \ground$, and suppose without loss of generality that $A$ and $B$ are disjoint. 
	Using Lemma~\ref{lemma:woodbury}, we show how to obtain a formula for $g(B \cup A) - g(A)$.
	Let us denote ${M_A = \Sigma^{-1} + \frac{1}{\sigma^2} X_A X_A^T}$. Using linearity and cyclic property of 
	trace, we obtain
	\begin{align*}
	g&(B \cup A) - g(A) \\
	&= \trinv{\Sigma^{-1} + \frac{1}{\sigma^2} X_{A} X_{A}^T } 
	- \trinv{ \Sigma^{-1} + \frac{1}{\sigma^2} X_{B \cup A} X_{B \cup A}^T }  \\
	&= \trinv{\Sigma^{-1} + \frac{1}{\sigma^2} X_{A} X_{A}^T } 
	- \trinv{ \Sigma^{-1} + \frac{1}{\sigma^2} X_{A} X_{A}^T + \frac{1}{\sigma^2} X_B X_B^T }  \\
	&= \trinv{M_A } 
	- \trinv{ M_A + \frac{1}{\sigma^2} X_B X_B^T}  \\
	&= \trinv{M_A} - 
	\tr{ M_A^{-1} - M_A^{-1} X_B \left( \sigma^2 I + X_B^T M_A^{-1} X_B \right)^{-1} X_B^T M_A^{-1}}
	&\text{(Lemma~\ref{lemma:woodbury})}\\
	&= \tr{M_A^{-1} X_B \left( \sigma^2 I + X_B^T M_A^{-1}  X_B \right)^{-1} X_B^T M_A^{-1}}
	\\
	&= \tr{\left( \sigma^2 I + X_B^T M_A^{-1}  X_B \right)^{-1} X_B^T M_A^{-2} X_B} \enspace,
	\end{align*}
	where the identity matrix is of size $|B|$. From this formula, we can easily derive the marginal gain of 
	a single element. In this case,  $B = \{e\}$ and  $X_B = x_e$, so the marginal gain is given by
	\begin{equation} \label{eq:a_opt_marginal_gain}
	g(e \mid A) = \frac{x_e^T M_A^{-2} x_e}{\sigma^2 + x_e^T M_A^{-1} x_e} \enspace.
	\end{equation}
	Note that $\Sigma^{-1} \preceq M_A$ (where $\preceq$ denotes the usual semidefinite ordering), 
	and thus, $M_A$ is positive definite. Hence, $M_A^{-1}$ and $M_A^{-2}$ are also positive definite, 
	which 
	means that their quadratic forms are non-negative. In particular, $x_e^T M_A^{-2} x_e \geq 0$ and 
	$x_e^T M_A^{-1} x_e \geq 0$, which implies $g(e \mid A) \geq 0$. Also note that $g(\varnothing)  = 
	0$. Combining this equality with the previous inequality, we get that $g$ is non-negative and 
	monotonically increasing. 
	
	Now we seek to show the lower bound on $\gamma$. Again, let $A, B \subseteq \ground$, and 
	assume 
	without loss of generality that $A$ and $B$ are disjoint. We seek to lower bound the ratio
	\begin{equation} \label{eq:sub_ratio}
	\frac{\sum_{e \in B} g(e \mid A) }{g(B \cup A) - g(A)} \enspace.
	\end{equation}
	Let $s = \max_{e \in \ground} \| x_e \|_2$. Observe that
	\begin{equation} \label{eq:denom_lb}
	\sigma^2 + x_e^T M_A^{-1} x_e 
	= \sigma^2 + \|x_e\|^2 \left( \frac{x_e^T M_A^{-1} x_e }{\|x_e\|^2} \right)
	\leq \sigma^2 + s^2 \lambda_{\max}\left( M_A^{-1} \right)
	= \sigma^2 + s^2 \lambda_{\max} \left( \Sigma \right) 
	\enspace,
	\end{equation}
	where the first inequality follows from the Courant-Fischer theorem, i.e., the variational 
	characterization of eigenvalues. 
	The second inequality is derived as follows: $M_A = \Sigma^{-1} + \frac{1}{\sigma} X_A X_A^T$ and 
	so $M_A \succeq \Sigma^{-1}$. This means that $M_A^{-1} \preceq \Sigma$. Thus, 
	$\lambda_{max}(M_A^{-1}) \leq \lambda_{max} \left( \Sigma \right)$.
	Using this, we may obtain a lower bound on the numerator in \eqref{eq:sub_ratio}.
	\begin{align*}
	\sum_{e \in B} g(e \mid A) 
	&= \sum_{e \in B} \frac{x_e^T M_A^{-2} x_e}{\sigma^2 + x_e^T M_A^{-1} x_e} 
	&\text{(by \eqref{eq:a_opt_marginal_gain})}\\
	&=\sum_{e \in B} \frac{ \tr{x_e x_e^T M_A^{-2}}}{\sigma^2 + x_e^T M_A^{-1} x_e} 
	&\text{(cyclic property of trace)}\\
	&\geq \frac{1}{\sigma^2 + s^2 \lambda_{\min} \left( M_A \right)} 
	\sum_{e \in B} \tr{x_e x_e^T M_A^{-2}}
	&\text{(by \eqref{eq:denom_lb})}\\
	&= \frac{\tr{X_B X_B^T M_A^{-2}}}{\sigma^2 + s^2 \lambda_{\min} \left( M_A \right)} 
	&\text{(linearity of trace)} \\
	&= \frac{\tr{X_B^T M_A^{-2} X_B}}{\sigma^2 + s^2 \lambda_{\min} \left( M_A \right)} \enspace.
	&\text{(cyclic property of trace)}
	\end{align*}
	Now, we will bound the denominator of \eqref{eq:sub_ratio}. We have already shown that
	\[
	g(B \cup A) - g(A) 
	= \tr{\left( \sigma^2 I + X_B^T M_A^{-1}  X_B \right)^{-1} X_B^T M_A^{-2} X_B} \enspace.
	\]
	Additionally, we have shown that $M_A^{-1}$ is positive semidefinite, and thus, $ X_B^T M_A^{-1}  
	X_B $ is also positive semidefinite. Hence,
	$\sigma^2 I \preceq \sigma^2 I + X_B^T M_A^{-1}  X_B$.
	This implies that $\left( \sigma^2 I + X_B^T M_A^{-1}  X_B \right)^{-1} \preceq \left( \sigma^2 I 
	\right)^{-1} = \frac{1}{\sigma^2} I$.
	Finally, we have shown that $M_A^{-2}$ is positive semidefinite, and therefore, we have that $ X_B^T 
	M_A^{-2}  
	X_B $ is also positive semidefinite. Thus,
	\[
	g(B \cup A) - g(A) 
	= \tr{\left( \sigma^2 I + X_B^T M_A^{-1}  X_B \right)^{-1} X_B^T M_A^{-2} X_B} 
	\leq \frac{1}{\sigma^2} \tr{ X_B^T M_A^{-2} X_B} \enspace.
	\]
	Applying these bound on $\sum_{e \in B} g(e \mid A)$ and $g(A \cup B) - g(A)$, we obtain
	\[
	\frac{\sum_{e \in B} g(e \mid A) }{g(B \cup A) - g(A)} 
	\geq \left( \frac{\sigma^2}{ \sigma^2 + s^2 \lambda_{max} (\Sigma)} \right) \frac{\tr{ X_B^T M_A^{-2} 
			X_B} 
	}{\tr{ X_B^T 
			M_A^{-2} X_B} }
	= \left( 1 + \frac{s^2}{\sigma^2} \lambda_{max} (\Sigma) \right)^{-1} \enspace.
	\qedhere
	\]
\end{proof}

Unlike submodular functions, lazy evaluations 
\citep{Minoux1978} of 
$\gamma$-weakly submodular $g$ are generally not possible, as the marginal gains vary unpredictably.
However, for specific functions, one can possibly speed up the greedy search.
For the utility $g$ considered here, we implemented a faster greedy search using 
the matrix 
inversion lemma.
The naive approach of computing $g(e \mid S)$ by constructing $\Sigma^{-1} + X_S X_S^T$, explicitly 
computing its inverse, and summing the diagonal elements is not only expensive{\textemdash}inversion 
alone costs 
$O(d^3)$ arithmetic operations{\textemdash}but also memory-inefficient. 
Instead, \eqref{eq:a_opt_marginal_gain} shows that 
\[ g(e \mid S) = \frac{\| z_e \|^2 }{\sigma^2 + \langle x_e, z_e \rangle} \enspace, \]
where $z_e = M_S^{-1} x_e$ and $M_S = \Sigma^{-1} + X_S X_S^T$.
In fact, $M_S^{-1}$ may be stored and 
updated directly in each iteration using the matrix inversion lemma so that \emph{no matrix inversion 
	are required}. Note that $M_{\varnothing}^{-1} = 
\Sigma$, which is an input parameter. By the matrix inversion lemma, 
\[
M_{S \cup e}^{-1} = M_S^{-1} - \frac{M_S^{-1} x_e x_e^T M_S^{-1}}{\sigma^2 + x_e^T M_S^{-1} x_e} 
\enspace,
\]
which takes $O(d^2)$ arithmetic operations.
Once $M_S^{-1}$ is known explicitly, computing $g(e \mid S)$ is simply matrix-vector multiplication on 
a fixed matrix.
We found that this greatly improved the efficiency of our code.

For this experiment, we used the Boston Housing dataset \citep{Harrison78}, a standard benchmark 
dataset containing $d=14$ attributes of $n=506$ Boston homes, including average number of rooms 
per dwelling, proximity to the Charles River, and crime rate per capita. 
We preprocessed the data by normalizing the features to have a zero mean and a standard deviation 
of $1$.
As there is no specified cost per measurement, we assigned costs proportionally to initial marginal 
gains in utility;
that is, $c_e = \alpha g(e)$ for some $\alpha \in [0,1]$.
We set $\sigma = 1/\sqrt{d}$, and
randomly generated a normal prior with covariance $\Sigma = A D A^T$, where $A$ is randomly 
chosen as $A_{i,j} \sim N(0,1)$ and $D$ is diagonal with $D_{i,i} = (i/d)^2$.
We choose not to use $\Sigma = \beta I$, as we found this causes $g$ to be nearly modular along 
solution paths, yielding it an easy problem instance for all algorithms and not a suitable benchmark.

\begin{figure*}[ht]
	\centering
	\hspace{0.21in} 
	\subfloat[]{\includegraphics[width=0.48 \textwidth]{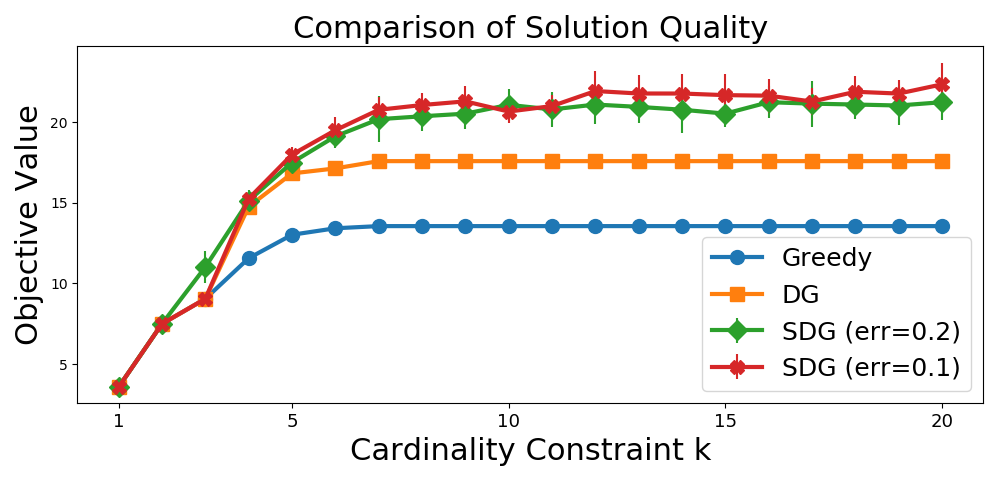}
		\label{fig:a_opt_vary_k}}
	\subfloat[]{\includegraphics[width=0.48 \textwidth]{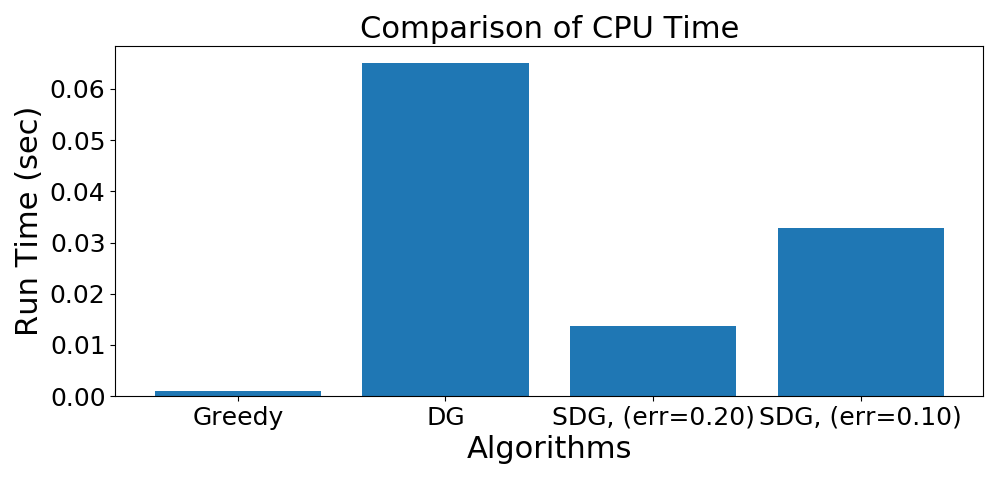}
		\label{fig:a_opt_timing}} \\
	\subfloat[]{\includegraphics[width=0.48 \textwidth]{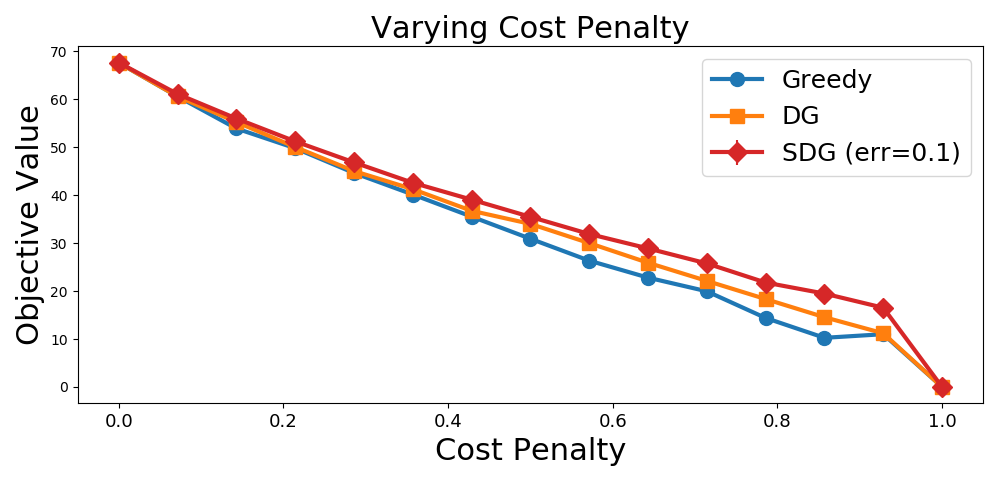}
		\label{fig:a_opt_vary_cost}}
	\subfloat[]{\includegraphics[width=0.48 \textwidth]{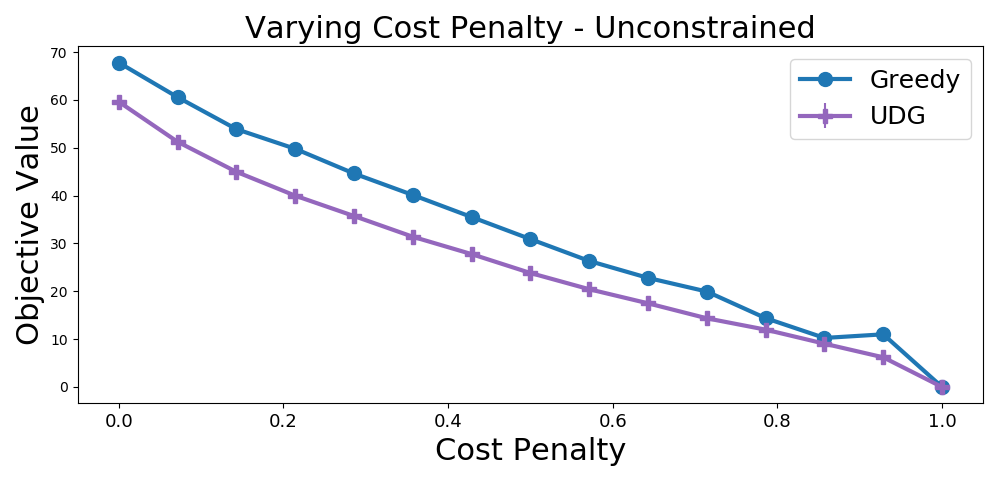}
		\label{fig:a_opt_vary_cost_unconstrained}} \\
	\caption{An algorithmic performance comparison for Bayesian $A$-Optimal design 
		on 
		the Boston Housing dataset.
		We report values for stochastic algorithms with mean and standard 
		deviation bars, over 20 trials.
		\eqref{fig:a_opt_vary_k} objective values, varying the 
		cardinality $k$, for a fixed cost penalty $\alpha = 0.8$.
		\eqref{fig:a_opt_timing}  runtime for a fixed cardinality $k=15$.
		\eqref{fig:a_opt_vary_cost} objective values, varying the cost penalty $\alpha$ 
		for 
		a fixed cardinality $k=15$.
		\eqref{fig:a_opt_vary_cost_unconstrained} objective values, varying the cost 
		penalty $\alpha$ in an unconstrained setting.
	}
	\label{fig:a_opt_results}
	\vspace{-0.1in}
\end{figure*}

In our first experiment, we fixed the cost penalty $\alpha = 0.8$, and ran the algorithms for varying 
cardinality constraints from $k=1$ to $k=15$. We ran the greedy algorithm, \dgreedy with \gsweep 
(setting 
$\delta = 0.1$), and two instances of \sdgreedy with \gsweep (with $\delta = \epsilon= 0.1$ 
and $\delta = \epsilon = 0.05$).
All \gsweep runs used $L=0$.
Figure~\ref{fig:a_opt_vary_k} compares the objective value of the sets returned by each of these  
algorithms. 
One can observe that the marginal gain obtained by the greedy algorithm is not non-increasing (at 
least for the first few elements), which is a result of the 
fact that $g$ is weakly submodular with $\gamma < 1$.
For small values of $k$, all algorithms produce comparable solutions; however, the greedy 
algorithm gets stuck in a local maximum of size $k=7$, while our algorithms are able to produce larger 
solutions with higher objective value. 
Moreover, \gsweep with \sdgreedy performs better than \gsweep with \dgreedy for larger values of $k$, 
for reasons discussed in Section~\ref{sec:gamma_sweep}.
Figure~\ref{fig:a_opt_timing} shows CPU times of each algorithm run with the single cardinality 
constraint  
$k=20$.
We see that the greedy algorithm runs faster than our algorithms. This difference in the runtime is a 
result of both the added complexity of the \gsweep procedure, 
and that greedy terminates early, when a local maximum is reached.
Figure~\ref{fig:a_opt_timing} also shows that the sub-sampling step in \sdgreedy results in a faster 
runtime than \dgreedy, as predicted by the theory.
We did not display the number of function evaluations, as it exhibits nearly identical trends to the 
actual CPU run time.
In our next experiment, we fixed the cardinality $k=15$ and varied the cost penalty $\alpha \in [0,1]$.
Figure~\ref{fig:a_opt_vary_cost} shows that all algorithm return similar solutions for $\alpha=0$ and 
$\alpha=1$, which are the cases in which either $c=0$ or the function $g - c$ is non-positive, respectively.
For all other values of $\alpha$, our algorithms yield improvements over greedy.
In our final experiment, we varied the cost penalty $\alpha \in [0,1]$, comparing the output of greedy 
and \gsweep with \udgreedy for the unconstrained setting.
Figure~\ref{fig:a_opt_vary_cost_unconstrained} shows that greedy outperforms our algorithm in this 
instance, which can occur, especially in the absence of ``bad elements'' of the kind discussed in 
Section~\ref{sec:algorithms}.

\subsection{Directed Vertex Cover with Costs}

The second experiment is directed vertex cover with costs. Let $G = (V, E)$ be a directed graph and 
let $w\colon V \rightarrow \reals$ be a weight function on the vertices. For a vertex set $S \subseteq V$, let 
$N(S)$ denote the set of vertices which are pointed to by $S$, $N(S) \triangleq \left\{ v \in V \mid (u,v) 
\in E \text{ for some } 
u \in S 
\right\}$. The weighted directed vertex cover function is 
$ g(S) = \sum_{u \in N(S) \cup S}w_u$.
We also assume that each vertex $v \in V$ has an associated nonnegative cost $c_v$. Our goal is to 
maximize the resulting revenue,
\[g(S) - c(S) = \sum_{u \in N(S) \cup S} \mspace{-18mu} w_u - \sum_{u \in S} c_u \enspace. \]
Because $g$ is submodular, we can forgo the \gsweep routine and apply our algorithms directly with
$\gamma=1$. 
Moreover, we implement lazy evaluations of $g$ in our code. 

For our experiments, we use the EU Email Core network, a directed graph generated using email data 
from a large European research institution \citep{yin2017, Leskovec2007}. 
The graph has 1k nodes and 25k directed edges, where nodes represent people and a directed edge 
from $u$ to $v$ means that an email was sent from $u$ to $v$.
We assign each node a weight of $1$.
Additionally, as there are no costs in the dataset, we assign costs in the following manner.
For a fixed $q$, we set $c(v) = 1 + \max\{d(v) - q, 0\}$, where $d(v)$ is the out-degree of $v$.
In this way, all vertices with out-degree larger than $q$ have the same initial marginal gain $g(v) - c(v)= 
q$.

\begin{figure}[htb!]
	\centering
	\hspace{0.21in} 
	\subfloat[]{\includegraphics[width = 0.48 \textwidth]{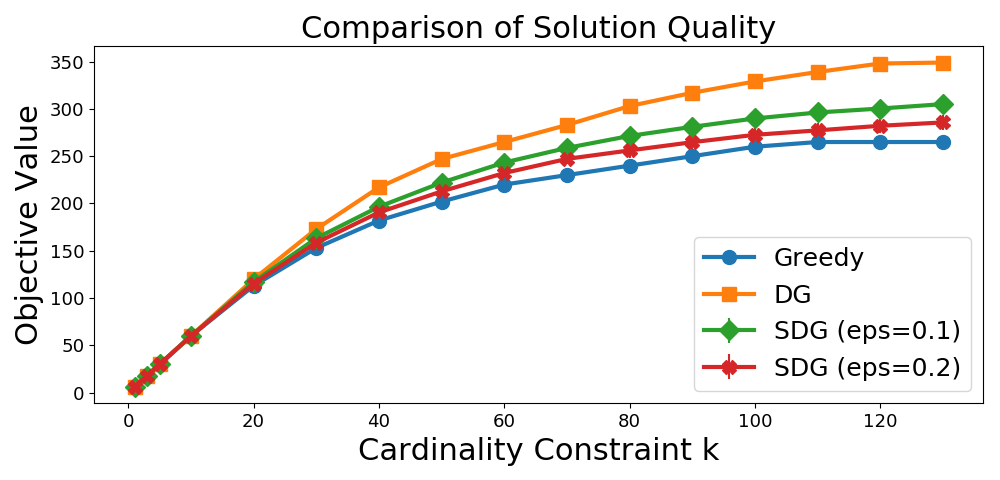}
		\label{fig:vc_fun_val}}
	\subfloat[]{\includegraphics[width = 0.48 \textwidth ]{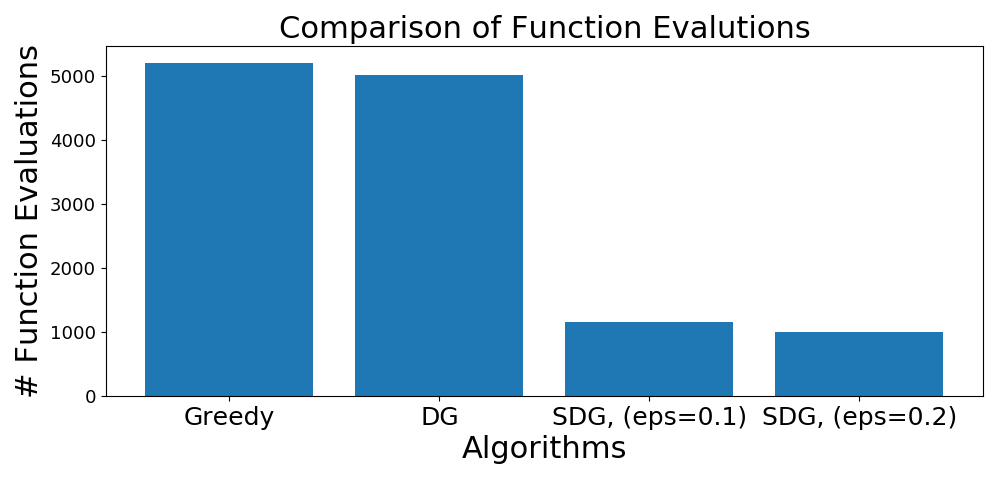}
		\label{fig:vc_fun_evals}} \\
	\subfloat[]{\includegraphics[width = 0.55 \textwidth]{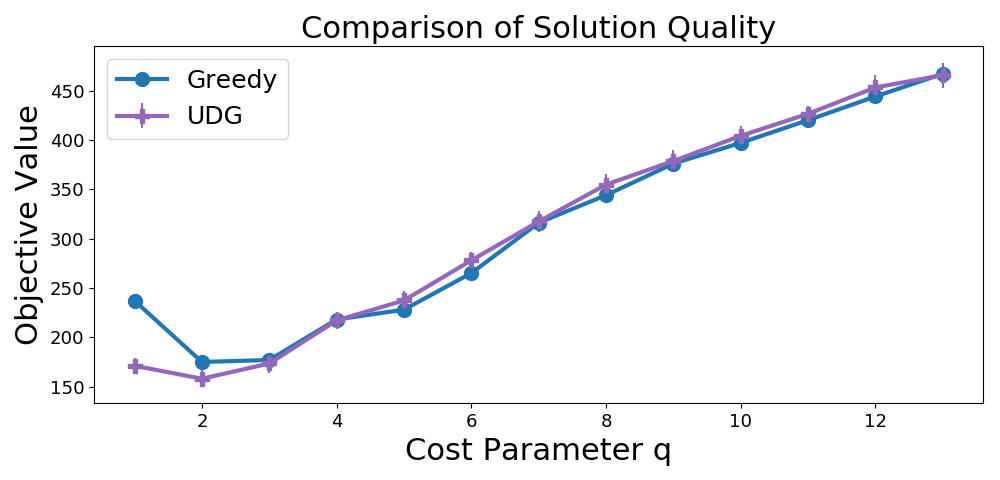}
		\label{fig:vc_unconstrained}}
	\caption{
		A performance comparison for directed vertex cover
		on the EU Email Core network.
		We report values for stochastic algorithms with mean and standard 
		deviation bars, over 20 trials.
		\eqref{fig:vc_fun_val} objective values, varying the 
		cardinality $k$, for a fixed cost factor $q = 6$.
		\eqref{fig:vc_fun_evals}  $g$ evaluations for a fixed cardinality $k=130$.
		\eqref{fig:vc_unconstrained} objective values, varying the cost factor $q$ in an unconstrained 
		setting.
	}
	\label{fig:vc_results}
	\vspace{-0.2in}
\end{figure}

In our first experiment, we fixed the cost factor $q=6$, and ran the algorithms for varying cardinality 
constraints from $k=1$ to $k=130$.
We see in Figure~\ref{fig:vc_fun_val} that our methods outperform greedy. 
\dgreedy achieves the highest objective value for each cardinality constraint, while \sdgreedy achieves 
higher objective values as the accuracy parameter $\epsilon$ is decreased.
Figure~\ref{fig:vc_fun_evals} shows the number of function evaluations made by the algorithms when 
$k=130$.
We observe that \sdgreedy requires much fewer function evaluations, even when lazy evaluations are 
implemented.%
\footnote{We do not report the CPU time for this experiment, as its behavior is somewhat different 
than the behavior of the number of function 
evaluations. This is an artifact of the implementation of the data structure we use to store the lazy evaluations.}
Finally, we ran greedy and \udgreedy while varying the cost factor $q$ from 1 to 12, and we note that in 
this setting (as can be seen in Figure~\ref{fig:vc_unconstrained}) our algorithm performs similarly to 
the greedy algorithm.

\section{Conclusion} \label{sec:conclusion}
We presented a suite of fast algorithms for maximizing the difference between a non-negative 
monotone $\gamma$-weakly submodular $g$ and a non-negative modular $c$ 
in both the cardinality constrained and unconstrained settings.
Moreover, we gave a matching hardness result showing that no algorithm can do better with only 
polynomially many oracle queries to $g$.
Finally, we experimentally validated our algorithms on Bayesian $A$-Optimality with costs and directed 
vertex cover with costs, and demonstrated that they outperform the greedy heuristic.


\bibliographystyle{plainnat}
\bibliography{icml2019}

\begin{thebibliography}{33}
\providecommand{\natexlab}[1]{#1}
\providecommand{\url}[1]{\texttt{#1}}
\expandafter\ifx\csname urlstyle\endcsname\relax
  \providecommand{\doi}[1]{doi: #1}\else
  \providecommand{\doi}{doi: \begingroup \urlstyle{rm}\Url}\fi

\bibitem[Bian et~al.(2017)Bian, Buhmann, Krause, and Tschiatschek]{Bian2017}
Andrew~An Bian, Joachim~M. Buhmann, Andreas Krause, and Sebastian Tschiatschek.
\newblock Guarantees for greedy maximization of non-submodular functions with
  applications.
\newblock In \emph{Proceedings of the 34th International Conference on Machine
  Learning}, volume~70, pages 498--507, 2017.

\bibitem[Boykov et~al.(2001)Boykov, Veksler, and Zabih]{Boykov2001}
Y.~Boykov, O.~Veksler, and R.~Zabih.
\newblock Fast approximate energy minimization via graph cuts.
\newblock \emph{IEEE Transactions on Pattern Analysis and Machine
  Intelligence}, 23\penalty0 (11):\penalty0 1222--1239, 2001.

\bibitem[Buchbinder and Feldman(2016)]{Buchbinder2016}
Niv Buchbinder and Moran Feldman.
\newblock Constrained submodular maximization via a non-symmetric technique.
\newblock \emph{CoRR}, abs/1611.03253, 2016.

\bibitem[Chamon and Ribeiro(2017)]{chamon2017}
Luiz F.~O. Chamon and Alejandro Ribeiro.
\newblock Approximate supermodularity bounds for experimental design.
\newblock In \emph{Advances in Neural Information Processing Systems}, 2017.

\bibitem[Chekuri et~al.(2014)Chekuri, Vondr{\'{a}}k, and
  Zenklusen]{Chekuri2014}
Chandra Chekuri, Jan Vondr{\'{a}}k, and Rico Zenklusen.
\newblock Submodular function maximization via the multilinear relaxation and
  contention resolution schemes.
\newblock \emph{{SIAM} J. Comput.}, 43\penalty0 (6):\penalty0 1831--1879, 2014.

\bibitem[Chv\'{a}tal(1979)]{C79}
V.~Chv\'{a}tal.
\newblock The tail of the hypergeometric distribution.
\newblock \emph{Discrete Mathematics}, 25\penalty0 (3):\penalty0 285--287,
  1979.

\bibitem[Das and Kempe(2011)]{das2011submodular}
Abhimanyu Das and David Kempe.
\newblock {Submodular meets Spectral: Greedy Algorithms for Subset Selection,
  Sparse Approximation and Dictionary Selection}.
\newblock In \emph{{International Conference on Machine Learning}}, pages
  1057--1064, 2011.

\bibitem[Elenberg et~al.(2017)Elenberg, Dimakis, Feldman, and
  Karbasi]{elenberg2017}
Ethan~R. Elenberg, Alexandros~G. Dimakis, Moran Feldman, and Amin Karbasi.
\newblock Streaming weak submodularity: Interpreting neural networks on the
  fly.
\newblock In \emph{Advances in Neural Information Processing Systems}, 2017.

\bibitem[Elenberg et~al.(2018)Elenberg, Khanna, Dimakis, and
  Negahban]{elenberg2018}
Ethan~R. Elenberg, Rajiv Khanna, Alexandros~G. Dimakis, and Sahand Negahban.
\newblock Restricted strong convexity implies weak submodularity.
\newblock \emph{Annals of Statistics}, 46, 2018.

\bibitem[Ene and Nguyen(2016)]{Ene2016}
Alina Ene and Huy~L. Nguyen.
\newblock Constrained submodular maximization: Beyond 1/e.
\newblock In \emph{FOCS}, pages 248--257, 2016.

\bibitem[Feldman(2019)]{feldman2019}
Moran Feldman.
\newblock Guess free maximization of submodular and linear sums, 2019.
\newblock To appear in WADS 2019.

\bibitem[Feldman et~al.(2017)Feldman, Harshaw, and Karbasi]{Feldman2017}
Moran Feldman, Christopher Harshaw, and Amin Karbasi.
\newblock Greed is good: Near-optimal submodular maximization via greedy
  optimization.
\newblock In \emph{COLT}, pages 758--784, 2017.

\bibitem[Golovin and Krause(2011)]{golovin11}
Daniel Golovin and Andreas Krause.
\newblock Adaptive submodularity: Theory and applications in active learning
  and stochastic optimization.
\newblock \emph{Journal of Artificial Intelligence Research}, 42:\penalty0
  427--486, 2011.

\bibitem[Hoeffding(1963)]{H63}
Wassily Hoeffding.
\newblock Probability inequalities for sums of bounded random variables.
\newblock \emph{Journal of the American Statistical Association}, 1963.

\bibitem[Hu et~al.(2016)Hu, Grubb, Bagnell, and Hebert]{Hu2016EfficientFG}
Hanzhang Hu, Alexander Grubb, J.~Andrew Bagnell, and Martial Hebert.
\newblock Efficient feature group sequencing for anytime linear prediction.
\newblock In \emph{Proceedings of the Thirty-Second Conference on Uncertainty
  in Artificial Intelligence}, 2016.

\bibitem[Jegelka and Bilmes(2011)]{jegelka2011submodularity}
Stefanie Jegelka and Jeff Bilmes.
\newblock {Submodularity beyond submodular energies: coupling edges in graph
  cuts}.
\newblock In \emph{Computer Vision and Pattern Recognition (CVPR)}, pages
  1897--1904. IEEE, 2011.

\bibitem[Jr. and Rubenfield(1978)]{Harrison78}
David~Harrison Jr. and Daniel~L Rubenfield.
\newblock Hedonic housing prices and the demand for clean air.
\newblock \emph{J. of Environmental Economics and Management}, 5\penalty0
  (1):\penalty0 81--102, 1978.

\bibitem[Kempe et~al.(2003)Kempe, Kleinberg, and Tardos]{kempe03}
David Kempe, Jon Kleinberg, and \'{E}va Tardos.
\newblock Maximizing the spread of influence through a social network.
\newblock In \emph{Proceedings of the ninth ACM SIGKDD international conference
  on Knowledge discovery and data mining}, pages 137--146. ACM, 2003.

\bibitem[Khanna et~al.(2017)Khanna, Elenberg, Dimakis, Negahban, and
  Ghosh]{khanna2017scalable}
Rajiv Khanna, Ethan~R. Elenberg, Alexandros~G. Dimakis, Sahand Negahban, and
  Joydeep Ghosh.
\newblock {Scalable Greedy Feature Selection via Weak Submodularity}.
\newblock In \emph{Proceedings of the 20th International Conference on
  Artificial Intelligence and Statistics ({AISTATS})}, pages 1560--1568, 2017.

\bibitem[Krause and Guestrin(2005)]{krause05near}
A.~Krause and C.~Guestrin.
\newblock {Near-optimal Nonmyopic Value of Information in Graphical Models}.
\newblock In \emph{Uncertainty in Artificial Intelligence (UAI)}, pages
  324--331, 2005.

\bibitem[Kuhnle et~al.(2018)Kuhnle, Smith, Crawford, and Thai]{KSCT18}
Alan Kuhnle, J.~David Smith, Victoria~G. Crawford, and My~T. Thai.
\newblock Fast maximization of non-submodular, monotonic functions on the
  integer lattice.
\newblock In \emph{ICML}, pages 2791--2800, 2018.

\bibitem[Lee et~al.(2010)Lee, Sviridenko, and Vondr{\'{a}}k]{Lee2010}
Jon Lee, Maxim Sviridenko, and Jan Vondr{\'{a}}k.
\newblock Submodular maximization over multiple matroids via generalized
  exchange properties.
\newblock \emph{Math. Oper. Res.}, 35\penalty0 (4):\penalty0 795--806, 2010.

\bibitem[Leskovec et~al.(2007)Leskovec, Kleinberg, and Faloutsos]{Leskovec2007}
Jure Leskovec, Jon Kleinberg, and Christos Faloutsos.
\newblock Graph evolution: Densification and shrinking diameters.
\newblock \emph{ACM Trans. Knowl. Discov. Data}, 1\penalty0 (1), 2007.

\bibitem[Lin and Bilmes(2011)]{lin2011class}
Hui Lin and Jeff Bilmes.
\newblock A class of submodular functions for document summarization.
\newblock In \emph{Proceedings of the 49th Annual Meeting of the Association
  for Computational Linguistics: Human Language Technologies-Volume 1}, pages
  510--520. Association for Computational Linguistics, 2011.

\bibitem[Minoux(1978)]{Minoux1978}
Michel Minoux.
\newblock Accelerated greedy algorithms for maximizing submodular set
  functions.
\newblock In \emph{Optimization Techniques}, pages 234--243, 1978.

\bibitem[Mirzasoleiman et~al.(2015)Mirzasoleiman, Badanidiyuru, Karbasi,
  Vondr{\'{a}}k, and Krause]{Mirzasoleiman2015}
Baharan Mirzasoleiman, Ashwinkumar Badanidiyuru, Amin Karbasi, Jan
  Vondr{\'{a}}k, and Andreas Krause.
\newblock Lazier than lazy greedy.
\newblock In \emph{Proceedings of the Twenty-Ninth {AAAI} Conference on
  Artificial Intelligence}, pages 1812--1818, 2015.

\bibitem[Nemhauser and Wolsey(1978)]{Nemhauser1978}
G~L Nemhauser and L~A Wolsey.
\newblock Best algorithms for approximating the maximum of a submodular set
  function.
\newblock \emph{Mathematics of Operations Research}, 3\penalty0 (3):\penalty0
  177--188, 1978.

\bibitem[Nemhauser et~al.(1978)Nemhauser, Wolsey, and Fisher]{Nemhauser1978a}
G~L Nemhauser, L~A Wolsey, and M~L Fisher.
\newblock An analysis of approximations for maximizing submodular set
  functions--{I}.
\newblock \emph{Mathematical Programming}, 14\penalty0 (1):\penalty0 265--294,
  1978.

\bibitem[Skala(2013)]{S13}
Matthew Skala.
\newblock Hypergeometric tail inequalities: ending the insanity.
\newblock \emph{CoRR}, abs/1311.5939, 2013.

\bibitem[Sviridenko(2004)]{Sviridenko04}
Maxim Sviridenko.
\newblock A note on maximizing a submodular set function subject to a knapsack
  constraint.
\newblock \emph{Oper. Res. Lett.}, 32\penalty0 (1):\penalty0 41--43, 2004.

\bibitem[Sviridenko et~al.(2017)Sviridenko, Vondr{\'{a}}k, and
  Ward]{Sviridenko2017}
Maxim Sviridenko, Jan Vondr{\'{a}}k, and Justin Ward.
\newblock Optimal approximation for submodular and supermodular optimization
  with bounded curvature.
\newblock \emph{Math. Oper. Res.}, 42\penalty0 (4):\penalty0 1197--1218, 2017.

\bibitem[Wei et~al.(2013)Wei, Liu, Kirchhoff, and Bilmes]{wei2013speech}
Kai Wei, Yuzong Liu, Katrin Kirchhoff, and Jeff Bilmes.
\newblock {Using Document Summarization Techniques for Speech Data Subset
  Selection}.
\newblock In \emph{Proceedings of NAACL-HLT 2013}, page 721–726, 2013.

\bibitem[Yin et~al.(2017)Yin, Benson, Leskovec, and Gleich]{yin2017}
Hao Yin, Austin~R. Benson, Jure Leskovec, and David~F. Gleich.
\newblock Local higher-order graph clustering.
\newblock In \emph{Proceedings of the 23rd ACM SIGKDD International Conference
  on Knowledge Discovery and Data Mining}. ACM, 2017.

\end{thebibliography}

\newpage 
\appendix
 
\section{Greedy Performs Arbitrarily Poorly} \label{sec:greedy_performs_poorly}

In this section, we describe an instance of Problem~\eqref{eq:main_problem} where the greedy algorithm 
performs arbitrarily poorly. More specifically, the greedy algorithm does not achieve any constant 
factor approximation.
Let $G$ be a graph with $n$ vertices and let $b \in V$ be a ``bad vertex''. The graph $G$ includes a single directed edge $(b, e)$ for every vertex $e \in V \setminus \{b\}$, and no other edges (i.e., $G$ is a directed star with $b$ in the center). Let $g$ be the 
unweighted directed vertex cover function. Note that
\begin{displaymath}
g(\{e\}) = \left\{
\begin{array}{lr}
1 &\text{ if } e \neq b \enspace,\\
n  &\text{ if } e = b \enspace.
\end{array}
\right.
\end{displaymath} 
Fix some $\epsilon > 0$, and let us define the nonnegative 
costs coefficients as
\begin{displaymath}
c_e = \left\{
\begin{array}{lr}
1/2 &\text{ if } e \neq b \enspace,\\
n - (1/2 + \epsilon) &\text{ if } e = b \enspace.
\end{array}
\right.
\end{displaymath} 
The initial marginal gain of a vertex $e$ is now given by 
\begin{displaymath}
g(\{e\}) - c_e = \left\{
\begin{array}{lr}
1/2 &\text{ if } e \neq b \enspace,\\
1/2 + \epsilon &\text{ if } e = b \enspace.
\end{array}
\right.
\end{displaymath} 
Thus, the greedy algorithm chooses the ``bad element'' $b \in V$ in the first iteration. Note that after 
$b$ is chosen, the greedy algorithm terminates, as $g(e \mid \{b\}) = 0$ and $c_e > 0$ for all remaining 
vertices $e$. However, for any set $S$ of vertices which does not contain $b$, we have that
\[ g(S) - c(S) = |S| - \frac{1}{2}|S| = \frac{1}{2} |S| \enspace. \]
Thus, for any $k < n$, the competitive ratio of greedy subject to a $k$ cardinality 
constraint is at most
\[\frac{1/2 + \epsilon}{k/2} = \frac{1+2\epsilon}{k} =  O \left( \frac{1}{k} \right) \enspace . \]

\end{document}